\documentclass[11pt]{article}

\usepackage{fullpage}
\usepackage{amsmath}
\usepackage{hyperref}
\usepackage{url}
\usepackage{algorithm}
\usepackage[noend]{algorithmic}
\usepackage{graphicx,tikz}
\usepackage{amsopn,amssymb,amsthm,amsmath}
\usepackage{xspace}
\usepackage{subfigure,color}

\newcommand{\R}{\mathbb{R}}

\newcommand{\svd}{\operatorname{SVD}}

\newcommand{\rank}{\operatorname{rank}}

\newcommand{\nnz}{\operatorname{nnz}}
\newcommand{\poly}{\operatorname{poly}}

\newcommand{\FD}{\operatorname{FD}}

\newcommand{\eps}{\varepsilon}
\newcommand{\rows}{\operatorname{rows}}

\newcommand{\fd}{\textsc{FrequentDirections}\xspace}
\newcommand{\sfd}{\textsc{SparseFrequentDirections}\xspace}
\newcommand{\SFD}{\textsc{SFD}\xspace}

\newcommand{\etal}{\emph{et al.}\xspace}
\newcommand{\si}{\textsc{SimultaneousIteration}\xspace}
\newcommand{\SSh}{\textsc{SparseShrink}\xspace}
\newcommand{\BSSh}{\textsc{BoostedSparseShrink}\xspace}
\newcommand{\DSh}{\textsc{DenseShrink}\xspace}

\newcommand{\vs}{\textsc{VerifySpectral}\xspace}

\newtheorem{theorem}{Theorem}[section]
\newtheorem{lemma}{Lemma}[section]

\newtheorem{fact}{Fact}[section]

\newcommand{\Paragraph}[1]{\paragraph*{\sffamily \textbf{#1}}}

\newcommand{\denselist}{\itemsep -2pt\parsep=-1pt\partopsep -2pt}

\newlength{\figsize} 

\title{Efficient Frequent Directions Algorithm for Sparse Matrices}
\author{
Mina Ghashami\\ University of Utah \\ \texttt{ghashami@cs.utah.edu}
\and
Edo Liberty\\ Yahoo Labs\\ \texttt{edo.liberty@yahoo.com}
\and
Jeff M. Phillips\thanks{Thanks to support by NSF CCF-1350888, IIS-1251019, ACI-1443046, and CNS-1514520.}\\ University of Utah \\ \texttt{jeffp@cs.utah.edu}
}
\date\nonumber

\begin{document}

\maketitle
\begin{abstract}
This paper describes Sparse Frequent Directions, a variant of Frequent Directions for sketching sparse matrices. 
It resembles the original algorithm in many ways:
both receive the rows of an input matrix $A^{n \times d}$ one by one in the streaming setting and compute a small sketch $B \in \R^{\ell \times d}$.
Both share the same strong (provably optimal) asymptotic guarantees with respect to the space-accuracy tradeoff in the streaming setting. However, unlike Frequent Directions which runs in $O(nd\ell)$ time regardless of the sparsity of the input matrix $A$, Sparse Frequent Directions runs in $\tilde{O}\left(\nnz(A)\ell + n\ell^2\right)$ time.
Our analysis loosens the dependence on computing the Singular Value Decomposition (SVD) as a black box within the Frequent Directions algorithm. Our bounds require recent results on the properties of fast approximate SVD computations.  
Finally, we empirically demonstrate that these asymptotic improvements are practical and significant on real and synthetic data.  
\end{abstract}

\section{Introduction}
\label{sec:intro}
It is very common to represent data in the form of a matrix. 
For example, in text analysis under the bag-of-words model, a large corpus of documents can be represented as a matrix whose rows refer to the documents and columns correspond to words. A non-zero in the matrix corresponds to a word appearing in the a document. 
Similarly, in recommendation systems~\cite{drineas2002competitive}, preferences of users are represented as a matrix with rows corresponding to users and columns corresponding to items. Non-zero entires correspond to user ratings or actions.

A large set of data analytic tasks rely on obtaining a low-rank approximation of the data matrix. 
These include clustering, dimension reduction, principal component analysis (PCA), signal denoising, etc. 
Such approximations can be computed using the Singular Value Decompositions (SVD).  
For an $n \times d$ matrix $A$ ($d\le n$) computing the SVD requires $O(n d^2)$ time and $O(nd)$ space in memory on a single machine. 
In many scenarios, however, data matrices are extremely large and computing their SVD exactly is infeasible. 
Efficient approximate solutions exist for distributed setting or when data access otherwise is limited.
In the row streaming model, the matrix rows are presented to the algorithm one by one in an arbitrary order.
The algorithm is tasked with processing the stream in one pass while being severely restricted in its memory footprint. 
At the end of the stream, the algorithm must provide a sketch matrix $B$ which is a good approximation of $A$ even though it is significantly more compact. This is called matrix sketching.

Matrix sketching methods are designed to be parallelizable, space and time efficient, and easily updatable.  
Computing the sketch on each machine and then combining the sketches together should be as good as sketching the combined data from all the different machines.
The streaming model is especially attractive since a sketch can be obtained and maintained as the data is being collected.
Therefore, eliminating the need for data storage altogether.

Often matrices, as above, are sparse;  most of their entries are zero. 
The work of \cite{dhillon2001concept} argues that typical term-document matrices are sparse; documents contain no more than $5\%$ of all words.  
On wikipedia, most words appear on only a small constant number of pages.    
Similarly, in recommendation systems, in average a user rates or interacts with a small fraction of the available items:  
less than $6\%$ in some user-movies recommendation tasks~\cite{asendorfalgorithms} and much fewer in physical purchases or online advertising.
As such, most of these datasets are stored as sparse matrices.  

There exist several techniques for producing low rank approximations of sparse matrices whose running time is 
$O(\nnz(A) \poly(k,1/\eps))$ for some error parameter $\eps \in (0,1)$. 
Here $\nnz(A)$ denotes the number of non-zeros in the matrix $A$.
Examples include the power method~\cite{golub2012matrix}, random projection techniques~\cite{sarlos2006improved}, projection-hashing~\cite{clarkson2013low}, and instances of column selection techniques~\cite{drineas2006fast2}. 

However, for a recent and popular technique \fd (best paper of KDD 2013~\cite{liberty2013simple}), 
there is no known way to take advantage of the sparsity of the input matrix.  
While it is deterministic and its space-error bounds are known to be optimal for dense matrices in the row-update model~\cite{ghashami2015frequent}, it runs in $O(nd \ell)$ time to produce a sketch of size $\ell \times d$.  
In particular, it maintains a sketch with $\ell$ rows and updates it iteratively over a stream, periodically invoking a full SVD which requires $O(d \ell^2)$ time.  
Reliance on exact SVD computations seems to be the main hurdle in reducing the runtime to depend on $O(\nnz(A))$.  
This paper shows a version of \fd whose runtime depends on $O(\nnz(A))$. This requires a new understanding and a more careful analysis of \fd.  
It also takes advantage of block power methods (also known as Subspace Iteration, Simultaneous Iteration, or Orthogonal Iteration) that run in time proportional to $\nnz(A)$ but incur small approximation error~\cite{musco2015stronger}.  

\subsection{Linear Algebra Notations}

Throughout the paper we identify an $n \times d$ matrix $A$ with a set of $n$ rows $[a_1; a_2; \ldots; a_n]$ where each $a_i$ is a vector in $\R^d$.
The notation $a_i$ stands for the $i$th row of the matrix $A$.  By $[A ; a]$ we mean the row vector $a$ appended to the matrix $A$ as its last row. 
Similarly, $[A;B]$ stands for stacking two matrices $A$ and $B$ vertically.
The matrices $I_n$ and $0^{n \times d}$ denote the $n$-dimensional identity matrix and the full zero matrix of dimension $n \times d$ respectively. 
The notation $\mathcal{N}(0,1)^{d\times \ell}$ denotes the distribution over $d\times \ell$ matrices whose entries are drawn independently from the normal distribution $\mathcal{N}(0,1)$.
For a vector $x$ the notation $\|\cdot\|$ refers to the Euclidian norm $\|x\| = (\sum_i{x_i^2})^{1/2}$.
The Frobenius norm of a matrix $A$ is defined as $\|A\|_F = \sqrt{\sum_{i=1} \|a_i\|^2}$, and the operator (or spectral) norm of it is 
$\|A\|_2 = \sup_{x\ne 0}\|Ax\|/\|x\|$.

The notation $\nnz(A)$ refers to the number of non-zeros in $A$, and $\rho = \nnz(A) / (nd)$ denotes relative density of $A$.
The Singular Value Decomposition of a matrix $A \in \mathbb{R}^{m \times d}$ for $m \le d$ is denoted by $[U,\Sigma ,V] = \svd(A)$.
It guarantees that $A = U\Sigma V^T$, $U^TU = I_m$, $V^TV = I_m$, $U\in \R^{m \times m}$, $V\in \R^{d \times m}$,
and $\Sigma \in \R^{m \times m}$ is a non-negative diagonal matrix such that $\Sigma_{i,i} = \sigma_i$ and $\sigma_{1}\ge \sigma_{2} \ge \ldots \ge \sigma_{m} \ge 0$. 
It is convenient to denote by $U_k$, and $V_k$ the matrices containing the first $k$ columns of $U$ and $V$ and by $\Sigma_k \in \R^{k \times k}$ the top left $k \times k$ block of $\Sigma$.
The matrix $A_k = U_k \Sigma_k V_k^T$ is the best rank $k$ approximation of $A$ in the sense that $A_k = {\arg \min}_{C : \rank(C) \leq k} \|A - C\|_{2,F}$. In places where we use $\svd(A,k)$, we mean rank $k$ SVD of $A$.

The notation $\pi_B(A)$ denotes the projection of the rows of $A$ on the span of the rows of $B$. In other words, $\pi_B(A) = A B^\dagger B$ where $(\cdot)^\dagger$ indicates taking the Moore-Penrose psuedoinverse.
Alternatively, setting $[U,\Sigma ,V] = \svd(B)$, we have $\pi_B(A) = AVV^T$. We also denote $\pi_B^k(A) = AV_k V_k^T$, the right projection of $A$ on the top $k$ right singular vectors of $B$. 


\section{Matrix Sketching Prior Art}
This section reviews only matrix sketching techniques that run in input sparsity time and whose output sketch is independent of the number of rows in the matrix.  We categorize all known results into three main approaches (1) column/row subset selection (2) random projection based techniques and (3) iterative sketching techniques. 


\Paragraph{Column selection techniques} 
These techniques, which are also studied under the Column Subset Selection Problem (CSSP) in literature~\cite{frieze2004fast,drineas2003pass,boutsidis2009improved,deshpande2006adaptive,drineas2011faster,boutsidis2011near}, form the sketch $B$ by selecting a subset of ``important'' columns of the input matrix $A$. They maintain the sparsity of $A$ and make the sketch $B$ to be more interpretable. 
These methods are not typically streaming, nor running in input sparsity time. The only method of this group which achieves both is~\cite{drineas2006fast2} by Drineas \etal that uses reservoir sampling to become streaming. 
They select $O(k/\eps^2)$ columns proportional to their squared norm and achieve the Frobenius norm error bound $\|A-\pi_{B_k}(A)\|_F^2 \leq \|A-A_k\|_F^2 + \eps \|A\|_F^2$ with time complexity of $O((k^2/\eps^4) (d+k/\eps^2) + \nnz(A))$.  
In addition, they show that the spectral norm error bound $\|A-\pi_{B_k}(A)\|_2^2 \leq \|A-A_k\|_2^2 + \eps \|A\|_F^2$ holds if one selects $O(1/\eps^2)$ columns. 
Rudelson \etal~\cite{rudelson2007sampling} improved the latter error bound to $\|A-\pi_{B_k}(A)\|_2^2 \leq \|A-A_k\|_2^2 + \eps \|A\|_2^2$ by selecting $O(r/\eps^4 \log{(r/\eps^4)})$ columns, where $r = \|A\|_F^2 / \|A\|_2^2$ is the numeric rank of $A$. Note that in the result by~\cite{drineas2006fast2}, one would need $O(r^2/\eps^2)$ columns to obtain the same bound.

Another similar line of work is the CUR factorization~\cite{boutsidis2014optimal,drineas2003pass,drineas2006fast,drineas2008relative,mahoney2009cur} where methods select $c$ columns and $r$ rows of $A$ to form matrices $C\in\R^{n\times c}$, $R\in\R^{r\times d}$ and $U\in\R^{c\times r}$, and constructs the sketch as $B = CUR$. The only instance of this group that runs in input sparsity time is~\cite{boutsidis2014optimal} by Boutsidis and Woodruff, where they select $r = c = O(k/\eps)$ rows and columns of $A$ and construct matrices $C, U$ and $R$ with $\rank(U) = k$ such that with constant probability $\|A-CUR\|_F^2 \leq (1+\eps)\|A-A_k\|_F^2$. Their algorithm runs in $O(\nnz(A) \log n + (n+d) \poly(\log n, k, 1/\eps))$ time.


\Paragraph{Random projection techniques} 
These techniques~\cite{papadimitriou1998latent,vempala2004random,sarlos2006improved,liberty2007randomized} operate data-obliviously and maintain a $r \times d$ matrix $B = S A$ using a $r \times n$ random matrix $S$ which has the Johnson-Lindenstrauss Transform (JLT) property~\cite{matouvsek2008variants}.  
Random projection methods work in the streaming model, are computationally efficient, and sufficiently accurate in practice~\cite{desai2015improved}.
%
The state-of-the-art method of this approach is by Clarkson and Woodruff~\cite{clarkson2013low} which was later improved slightly in~\cite{NN13}. 
It uses a hashing matrix $S$ with only one non-zero entry in each column.  Constructing this sketch takes only  $O(\nnz(A) + n \cdot \poly(k/\eps) + \poly(dk/\eps))$ time, and guarantees that for any unit vector $x$ that 
$
(1-\eps)\|Ax\| \leq \|B x\| \leq (1+\eps) \|Ax\|.
$
For these sparsity-efficient sketches using $r = O(d^2/\eps^2)$ also guarantees that 
$\|A - \pi_B(A)\|_F \leq (1+\eps) \|A - A_k\|_F$.    

\Paragraph{Iterative sketching techniques} 
These operate in streaming model, where random access to the matrix is not available. 
They maintain the sketch $B$ as a linear combination of rows of $A$, and update it as new rows are received in the stream.  
Examples of these methods include different version of iterative SVD~\cite{golub2012matrix, hall1998incremental, levey2000sequential, brand2002incremental,ross2008incremental}. These, however, do not have theoretical guarantees~\cite{desai2015improved}.   
The \fd algorithm~\cite{liberty2013simple} is a unique in this group in that it offers strong error guarantees. 
It is a deterministic sketching technique that processes rows of an $n \times d$ matrix $A$ in a stream and maintains a $\ell \times d$ sketch $B$ (for $\ell < \min(n,d)$) such that the following two error bounds hold for any $0 \leq k < \ell$
\[
\|A^TA - B^TB\|_2 \leq \frac{1}{\ell-k}\|A-A_k\|_F^2
\]
and
\[
 \|A-\pi_B(A)\|_F^2 \leq \frac{\ell}{\ell-k} \|A-A_k\|_F^2 \ .
\]
Setting $\ell = k + 1/\eps$ and $\ell = k + k/\eps$, respectively, achieves bounds 
$
\|A^TA - B^TB\|_2 \leq \eps \|A-A_k\|_F^2
$
and
$
\|A-\pi_B(A)\|_F^2 \leq (1+\eps) \|A-A_k\|_F^2.
$
Although \fd does not run in input sparsity time, we will explain it in detail in the next section, as it is an important building block for the algorithm we introduce.  

\subsection{Main Results} 
We present a randomized version of \fd, called as \sfd that receives an $n \times d$ sparse matrix $A$ as a stream of its rows. 
It computes a $\ell \times d$ sketch $B$ in $O(d\ell)$ space. 

It guarantees that with probability at least $1-\delta$ (for $\delta \in (0,1)$ being the failure probability), for $\alpha = 6/41$ and any $0 \leq k < \alpha \ell$,
\[
\|A^T A - B^T B\|_2 \leq \frac{1}{\alpha \ell - k} \|A -A_k\|_F^2
\]
and
\[
\|A - \pi_{B_k}(A)\|_F^2 \leq \frac{\ell}{\ell - k/\alpha} \|A - A_k\|_F^2.
\] 
Note that setting $\ell = \lceil 1/(\eps\alpha) + k/\alpha\rceil$ yields 
\[
\|A^T A - B^T B\|_2 \leq \eps \|A -A_k\|_F^2
\]
and setting $\ell = \lceil k/(\eps\alpha) + k/\alpha\rceil$ yields 
\[
\|A - \pi_{B_k}(A)\|_F^2 \leq (1+\eps) \|A - A_k\|_F^2.
\]

The expected running time of the algorithm is
\[
O( \nnz(A)\ell \log(d) +  \nnz(A)\log (n/\delta)+ n\ell^2 + n\ell\log(n/\delta) ).
\] 
In the likely case where $\nnz(A) = \Omega (n \ell)$ and $n/\delta  < d^{O(\ell)}$, the runtime is dominated by $O(\nnz(A) \ell \log(d))$.  
We also experimentally validate this theory, demonstrating these runtime improvements on sparse data without sacrificing accuracy.

\section{Preliminaries}

In this section we review some important properties about \fd and \si which will be necessary for understanding and proving bounds on  \sfd.  

\subsection{Frequent Directions}
\label{related:fd}

The \fd algorithm was introduced by Liberty~\cite{liberty2013simple} and received an improved analysis by Ghashami et al.\cite{ghashami2015frequent}. 
The algorithm operates by collecting several rows of the input matrix and letting the sketch grow.
Once the sketch doubles in size, a lossy \textsc{DenseShrink} operation reduces its size by a half.
This process repeats throughout the stream.
The running time of \fd and its error analysis are strongly coupled with the properties of the 
SVD used to perform the \textsc{DenseShrink} step.

An analysis of~\cite{desai2015improved} slightly generalized the one in~\cite{ghashami2015frequent}.
Let $B$ be the sketch resulting in applying \fd with a potentially different shrink operation to $A$. 
Then, the \fd asymptotic guarantees hold as long as the shrink operation exhibits three properties, for any positive $\Delta$ and a constant $\alpha \in (0,1)$.
\begin{enumerate} 
\item Property 1:  For any vector $x\in\R^d$, $\|Ax\|^2 - \|Bx\|^2 \geq 0$.
\item Property 2: For any unit vector $x\in\R^d$, $\|Ax\|^2 - \|Bx\|^2 \leq \Delta$.
\item Property 3: $\|A\|_F^2 - \|B\|_F^2 \geq \alpha \Delta \ell$.
\end{enumerate}
For completeness, the exact guarantee is stated in Lemma~\ref{lem:desai}.

\begin{lemma}[Lemma 3.1 in~\cite{desai2015improved}]\label{lem:desai}
Given an input $n \times d$ matrix $A$ and an integer parameter $\ell$, any sketch $\ell \times d$ matrix $B$ which satisfies the three properties above (for some any $\alpha \in (0,1]$ and $\Delta > 0$),
guarantees the following error bounds
\[
0 \leq \|A^T A - B^T B\|_2 \leq \frac{1}{\alpha \ell-k} \|A -A_k\|_F^2, 
\]
and
\[
\|A - \pi_{B}^k(A)\|_F^2 \leq \frac{\ell}{ \ell -k/\alpha} \|A - A_k\|_F^2, 
\]
where $\pi^k_{B}(\cdot)$ represents the projection operator onto $B_k$, the top $k$ singular vectors of $B$.  
\end{lemma}

Another important property of \fd is that its sketches are mergeable~\cite{ghashami2015frequent}. To clarify, consider partitioning a matrix $A$ into $t$ blocks $A_1, A_2, \cdots, A_t$ so that $A = [A_1; A_2; \cdots; A_t]$. Let $B_i = \FD(A_i, \ell)$ denotes the $\FD$ sketch of the matrix block $A_i$, and $B' = \FD ([B_1; B_2; \cdots; B_t], \ell)$ denotes the $\FD$ sketch of all $B_i$s combined together.
It is shown that $B'$ has at most as much covariance and projection error as $B = \FD(A, \ell)$, i.e. the sketch of the whole matrix $A$. It follows that this divide-sketch-and-merging can also be applied recursively on each matrix block without increasing the error.

The runtime of \fd is determined by the number of shrinking steps. 
Each of those computes an $\svd$ of $B$ which takes $O(d\ell^2)$ time. 
Since the SVD is called only every $O(\ell)$ rows this yields a total runtime $O(d\ell^2 \cdot n/\ell) = O(nd\ell)$. This effectively means that on average we are spending $O(d\ell)$ operations per row, even if the row is sparse. 

In the present paper, we introduce a new method called as \sfd that uses randomized SVD methods instead of the exact SVD to approximate the singular vectors and values of intermediate matrices $B$.  We show how this new method tolerates the extra approximation error and runs in time proportional to $\nnz(A)$. Moreover, since it received sparse matrix rows, it can observe more the $\ell$ rows until the size of the sketch doubles.
As a remark, Ghashami and Phillips~\cite{ghashami2014relative} showed that maintaining any rescaled set of $\ell$ rows of $A$ over a stream is not a feasible approach to obtain sparsity in \fd. It was left as an open problem to produce some version of \fd that took advantage of the sparsity of $A$.  

\subsection{Simultaneous  Iteration}
Efficiently computing the singular vectors of matrices is one of the most well studies problems in scientific computing.
Recent results give very strong approximation guarantees for block power method techniques \cite{rokhlin2009randomized}\cite{woolfe2008fast}\cite{liberty2007randomized}\cite{halko2011finding}. Several variants of this algorithm were studied under different names in the literature e.g. Simultaneous  Iteration, Subspace Iteration, or Orthogonal Iteration~\cite{golub2012matrix}. 
In this paper, we refer to this group of algorithms collectively as \si. 
A generic version of \si for rectangular matrices is described in Algorithm~\ref{alg:sim_itr}.
\begin{algorithm}[H]
\caption{\si}
\label{alg:sim_itr}
\begin{algorithmic}
\STATE \textbf{Input}: $A \in \R^{n \times d}$, rank $k \leq \min(n,d)$, and error $\eps\in(0,1)$
\STATE $q = \Theta(\log (n/\eps)/\eps)$
\STATE $G \sim \mathcal{N}(0,1)^{d\times k}$
\STATE $Z = \operatorname{GramSchmidt}(A(A^TA)^qG)$
\STATE \textbf{return} $Z$ \hfill \# $Z\in\R^{n\times k}$
\end{algorithmic}
\end{algorithm}

While this algorithm was already analyzed by \cite{golub2012matrix}, the proofs of~\cite{rokhlin2009randomized, halko2011finding, musco2015stronger, witten2013randomized} manage to prove stable results that hold for any matrix independent of spectral gap issues.
Unfortunately, an in depth discussion of these algorithms and their proof techniques is beyond the scope of this paper. 

For the proof of correctness of \sfd, the main lemma proven by \cite{musco2015stronger} suffices.
\si (Algorithm~\ref{alg:sim_itr}) guarantees the three following error bounds with high probability:
\begin{enumerate}
\item Frobenius norm error bound: 
$\|A-ZZ^TA\|_F \leq (1+\eps) \|A-A_k\|_F$
\item Spectral norm error bound: 
$\|A-ZZ^TA\|_2 \leq (1+\eps) \|A-A_k\|_2$
\item Per vector error bound: 
$|u_i^T AA^T u_i - z_i^T AA^T z_i| \leq \eps \sigma_{k+1}^2$

for all $i$. Here $u_i$ denotes the $i$th left singular vector of $A$, and $\sigma_{k+1}$ is the ($k+1$)th singular value of $A$, and $z_i$ is the $i$th column of the matrix $Z$ returned by \si.  
\end{enumerate}

In addition, for a constant $\eps$, \si runs in $\tilde{O}(\nnz(A))$ time.

In this paper, we show that \sfd can replace the computation of an exact SVD by using the results of~\cite{musco2015stronger} with $\eps$ being a constant.  This alteration does give up the optimal asymptotic accuracy (matching that of \fd).

\section{Sparse Frequent Directions}
The \sfd (\SFD) algorithm is described in Algorithm \ref{alg:SparseFD}, and is an extension of \fd to sparse matrices.
It receives the rows of an input matrix $A$ in a streaming fashion and maintains a sketch $B$ of $\ell$ rows. Initially $B$ is empty. 
On receiving rows of $A$, \SFD stores non-zeros in a buffer matrix $A'$. 
The buffer is deemed full when it contains $\ell d$ non-zeros or $d$ rows. 
\SFD then calls \BSSh to produce its sketch matrix $B'$ of size $\ell \times d$.
Then, it updates its ongoing sketch $B$ of the entire stream by merging it with the (dense) sketch $B'$ using \DSh.
\begin{algorithm}[H]
\caption{\sfd}
\label{alg:SparseFD}
\begin{algorithmic}
\STATE \textbf{Input:} $A \in \R^{n \times d}$, an integer $\ell \le d$, failure probability $\delta$
\STATE $B = 0^{\ell \times d}$, \; $A' = 0^{0 \times d}$
\FOR {$a \in A$}
\STATE $A' = [A' ; a]$
\IF {$\nnz(A') \ge \ell d$ {\bf or} $\rows(A') = d$} 
\STATE $B' = \BSSh(A',\ell,\delta)$
\STATE $B =  \DSh([B; B'],\ell)$
\STATE $A' = 0^{0 \times d}$
\ENDIF
\ENDFOR
\STATE \textbf{return} $B$ 
\end{algorithmic}
\end{algorithm}

\BSSh amplifies the success probability of another algorithm \SSh in Algorithm \ref{alg:SparseShrink}.   \SSh runs \si instead of a full SVD to take advantage of the sparsity of its input $A'$.  However, as we will discuss, by itself \SSh has too high of a probability of failure.  
Thus we use \BSSh which keeps running \SSh and probabilistically verifying the correctness of its result using \vs, until it decides that the result is correct with high enough probability. 
Each of \DSh, \SSh, and \BSSh produce sketch matrices of size $\ell \times d$.


\begin{algorithm}[H]
\begin{algorithmic}
\caption{\textsc{SparseShrink}}
\label{alg:SparseShrink}
\STATE \textbf{Input}: $A' \in \R^{m \times d}$, an integer $ \ell \le m$

  \STATE $Z = \si(A',\ell,1/4)$ 
  \STATE $P = Z^TA'$, \; $[H,\Lambda, V] = \svd(P,\ell) $ 
  \STATE $\tilde \Lambda = \sqrt{\Lambda^2 - \lambda_{\ell}^2 I_\ell}$
  \STATE  $B' = \tilde \Lambda V^T$
  \STATE \textbf{return} $B'$
\end{algorithmic}
\end{algorithm}


\begin{algorithm}[H]
\begin{algorithmic}
\caption{\BSSh}
\label{alg:BoostedSparseShrink}
\STATE \textbf{Input}: $A' \in \R^{m \times d}$, integer $ \ell \le m$, failure probability $\delta$
\WHILE{True}
  \STATE $B' = \SSh(A',\ell)$
  \STATE $\Delta = (\|A'\|_F^2 - \|B'\|_F^2) / \alpha \ell$ \;\;\; for $\alpha = 6/41 $
  \IF {$\textsc{VerifySpectral}((A'^TA' - B'^T B')/(\Delta/2), \delta)$}
    \STATE \textbf{return} $B'$
  \ENDIF
\ENDWHILE
\end{algorithmic}
\end{algorithm}

\begin{algorithm}[H]
\caption{\textsc{DenseShrink}}
\label{alg:DenseFD}
\begin{algorithmic}
\STATE \textbf{Input}: $A \in \R^{m \times d}$, an integer $\ell \le m$
\STATE $[H,\Lambda,V] = \svd(A,\ell)$ 
\STATE $\tilde \Lambda = \sqrt{\Lambda^2 - \lambda^2_\ell I_\ell}$
\STATE $B = \tilde \Lambda V^T$
\STATE \textbf{Return} $B$ 
\end{algorithmic}
\end{algorithm}

Our main result is stated in the next theorem. It follows from combining the proofs contained in the subsections below.
\begin{theorem}[main result]\label{thm:bounds}
Given a sparse matrix $A\in\R^{n\times d}$ and an integer $\ell \le d$, \sfd computes a small sketch $B \in \R^{\ell \times d}$ such that with probability at least $1-\delta$ for $\alpha = 6/41$ and any $0 \leq k < \alpha \ell$,
\[
\|A^T A - B^T B\|_2 \leq \frac{1}{\alpha \ell - k} \|A -A_k\|_F^2 
\]
and
\[
\|A - \pi_{B_k}(A)\|_F^2 \leq \frac{\ell}{\ell - k/\alpha} \|A - A_k\|_F^2.
\]
The total memory footprint of the algorithm is $O(d\ell)$ and its expected running time is 
\[
O\left( \nnz(A)\ell \log(d) +  \nnz(A)\log (n/\delta)+ n\ell^2 + n\ell\log(n/\delta) \right). 
\]
\end{theorem}

\subsection{Success Probability}
\label{sec:suc-prob}
\SSh, described in Algorithm \ref{alg:SparseShrink}, calls \si to approximate the top rank $\ell$ subspace of $A'$.
As \si is randomized, it fails to converge to a good subspace when the initial choice of the random matrix $G$ does not sufficiently align with the top $\ell$ singular vectors of $A'$ (see Algorithm \ref{alg:sim_itr}). This occurs with probability at most $\rho_\ell = O(1/\sqrt{\ell})$.
In Section \ref{subsub:SSh}, we prove that with probability of at least $1-\rho_\ell$ that \SSh satisfies the three properties required for Lemma \ref{lem:desai} using $\alpha = 6/41$ and $\Delta = 41/8\; s_\ell^2$, but replacing Property $2$ with a stronger version
\begin{itemize}
\item Property 2 (strengthened):  
$\|A'^TA' - B'^T B'\|_2 \leq (\Delta/2) = 41/16\; s_\ell^2$
\end{itemize}
where $s_\ell$ denotes the $\ell$th singular value of $A'$.

However, for the proof of \sfd we require that {\it all} \SSh runs to be successful. 
The failure probability of \SSh, which is upper bounded by $O(1/\sqrt \ell)$, is high enough that a simple union bound would not give a meaningful bound on the failure probability of \sfd. We therefore reduce the failure probability of each \BSSh, by wrapping each call of \SSh in the verifier \vs.  If \vs does not verify the correctness, then it reruns \SSh and tries again until it can verify it.   
But to perform this verification efficiently, we need to loosen the definition of correctness.  In particular, we say \SSh is successful if the sketch $B'$ computed from its output satisfies $\|A'^T A' - B'^T B'\|_2 \leq \Delta$ (the original Property 2 specification in Section \ref{related:fd}), where $\Delta = (\|A'\|_F^2 - \|B'\|_F^2) / \alpha \ell$.  Combining the two inequalities through $\Delta$, a successful run implies that $\|A'^T A' - B'^TB'\|_2 \leq (\|A'\|_F^2 - \|B'\|_F^2) / \alpha \ell$.  
\vs verifies the success of the algorithm by approximating the spectral norm of $(A'^T A' - B'^T B')/(\Delta/2)$; it does so by running the power method for $c \cdot \log(d/\delta_i)$ steps for some constant $c$.  

\begin{algorithm}[H]
\caption{\vs}
\label{alg:verifyspectral}
\begin{algorithmic}
\STATE \textbf{Initialization} persistent $i=0$ ($i$ retains its state between invocations of this method)
\STATE \textbf{Input}: Matrix $C\in\R^{d\times d}$, failure probability $\delta$
\STATE $i = i+1$ and $\delta_i = \delta/2i^2$
\STATE Pick $x$ uniformly at random from the unit sphere in $\R^d$.  
\STATE \textbf{if} $\|C^{c \cdot \log(d/\delta_i)} x\| \leq 1$ \textbf{return} \textsc{True}
\STATE \textbf{else} \textbf{return} \textsc{False}
\end{algorithmic}
\end{algorithm}
\begin{lemma} \label{lem:vs}
The $\vs$ algorithm returns \textsc{True} if $\|C\|_2 \le 1$. 
If $\|C\|_2 \ge 2$ it returns \textsc{False} with probability at least $1-\delta_i$.
\end{lemma}
\begin{proof}
If $\|C\|\le 1$ than $\|C^{c\cdot \log(d/\delta_i)} x\|  \le \|C\|^{c \cdot \log(d/\delta_i)} \|x\| \le 1$.
If $\|C\|\ge 2$, consider execution $i$ of the method. Let $v_1$ denote the top singular vector of $C$. Then  $\|C^{c\cdot\log(d/\delta_i)} x\| \ge |\langle v_1,x \rangle| 2^{c\cdot \log(d/\delta_i)} \ge 1$, for some constant $c$ as long as  $|\langle v_1,x \rangle| = \Omega( \operatorname{poly}(\delta_i/d))$. Let $\Phi(t')$ denote the density function of the random variable 
$t' = \langle v_1,x \rangle$. Then $\Pr[|\langle v_1,x \rangle| \le t] = \int_{-t}^{t}\Phi(t')dt'  \le 2t\Phi(0) = O(t\sqrt{d})$. Setting the failure probability to be at most $\delta_i$, we conclude that $|\langle v_1,x \rangle| = \Omega(\delta_i/\sqrt{d})$ with probability at least $1-\delta_i$. 
\end{proof}

Therefore, \vs fails with probability at most  $\delta_i$ during execution $i$. 
If any of \vs runs fail, \BSSh and hence \sfd potentially fail.
Taking the union bound over all invocations of \vs we obtain that \sfd fails with probability at most $\sum \delta_i \le \sum_{i=1}^{\infty} \delta/2 i^2 \le \delta$, hence it succeeds with probability at least $1-\delta$. 


\subsection{Space Usage and Runtime Analysis}\label{sec:runtime_analysis}
Throughout this manuscript we assume the constant-word-size model. 
Integers and floating point numbers are represented by a constant number of bits. 
Random access into memory is assumed to require $O(1)$ time.
In this model, multiplying a sparse matrix $A'$ by a dense vector requires $O(\nnz(A'))$ operations and storing $A'$ requires $O(\nnz(A'))$ bits of memory. 

\begin{fact}
The total memory footprint of \sfd is $O(d\ell)$.
\end{fact}
\begin{proof}
It is easy to verify that, except for the buffer matrix $A'$, the algorithm only manipulates $\ell \times d$ matrices; in particular, observe that the ($\rows(A') = d$) condition in \sfd ensures that $m = d$ in \SSh, and in \DSh also $m = 2\ell$.
Each of these $\ell \times d$ matrices clearly require at most $O(d\ell)$ bits of memory. 
The buffer matrix $A'$ contains at most $O(d\ell)$ non-zeros and therefore does not increase the space complexity of the algorithm.
\end{proof}

We turn to bounding the expected runtime of \sfd which is dominated by the cumulative running times of  \DSh and \BSSh.
Denote by $T$ the number of times they are executed. It is easy to verify $T \le \nnz(A)/d\ell + n/d$. 
Since \DSh runs in $O(d\ell^2)$ time deterministically, the total time spent by \DSh through $T$ iterations is $O(Td\ell^2) = O(\nnz(A)\ell + n\ell^2)$.

The running time of \BSSh is dominated by those of \SSh and \vs, and its expected number of iterations. 
Note that, in expectation, they are each executed on any buffer matrix $A'_i$ a small constant number of times because \vs succeeds with probability (much) greater than $1/2$. For asymptotic analysis it is identical to assuming they are each executed once.

Note that the running time of \SSh on $A'_i$ is $O(\nnz(A'_i)\ell \log(d))$. 
Since $\sum_i \nnz(A'_i) = \nnz(A)$ we obtain a total running time of $O(\nnz(A)\ell \log(d))$.
The $i$th execution of \vs requires $O(d\ell \log(d/\delta_i))$ operations. 
This, because it multiplies $A'^TA' - B'^TB'$ by a single vector $O(\log(d/\delta_i))$ times, and both $\nnz(A') \leq O(d \ell)$ and $\nnz(B') \leq d \ell$.
In expectation \vs is executed $O(T)$ times, therefore total running time of it is 
\begin{align*}
O(d\ell \sum_{i=1}^{O(T)} \log(d/\delta_i)) = O(d\ell \sum_{i=1}^{O(T)} \log(d i^2/\delta))
=
O(T d\ell \log(Td/\delta)) = O((\nnz + n\ell)\log(n/\delta)).  
\end{align*}

Combining the above contributions to the total running time of the algorithm we obtain Fact~\ref{lem:sparseshrink_runtime}.
\begin{fact}
\label{lem:sparseshrink_runtime}
Algorithm \sfd runs in expected time of 
\[
O(\nnz(A)\ell \log(d) + \nnz(A)\log(n/\delta) + n\ell^2 + n\ell \log(n/\delta)).  
\]
\end{fact}

\subsection{Error Analysis}
\label{sec:error}

We turn to proving the error bounds of Theorem \ref{thm:bounds}.  
Our proof is divided into three parts.  
We first show that \SSh obtains the three properties needed for Lemma \ref{lem:desai} with probability at least $1-\rho_\ell$, and with the constraint on Property 2 strengthed by a factor $1/2$.
Then we show how loosening Property 2 back to its original bound enables \BSSh to succeed with probability $1-\delta_i$ for some $\delta_i \ll \rho_\ell$. 
Finally we show that due to the mergeability of \fd~\cite{Lib12}, discussed in Section~\ref{related:fd}, the \sfd algorithm obtains the same error guarantees as \BSSh with probability $1-\delta$ for a small $\delta$ of our choice. 

In what follows, we mainly consider a single run of \SSh or \BSSh and let $s_\ell$ and $u_\ell$ denote the $\ell$th singular value and $\ell$th left singular vector of $A'$, respectively.

\subsubsection{Error Analysis: \SSh} 

Here we show that with probability at least $1-\rho_\ell$ that $B'$ computed from $\SSh(A',\ell)$ satisfies the three properties discussed in Section \ref{related:fd} required for Lemma \ref{lem:desai}.  
\begin{itemize} \denselist
\item Property 1: For any unit vector $x\in\R^d$, $\|A'x\|^2 - \|B'x\|^2 \geq 0$,
\item Property 2 (strengthened): For any unit vector $x\in\R^d$, $\|A'x\|^2 - \|B'x\|^2 \leq \Delta/2 = (41/16) s_\ell^2$, 
\item Property 3: $\|A'\|_F^2 - \|B'\|_F^2 \geq \ell \alpha \Delta = \ell (3/4) s_\ell^2 $.
\end{itemize}

\label{subsub:SSh}
\begin{lemma} 
Property 1 holds deterministically for \SSh: $\|A' x\|^2 - \|B' x\|^2 \ge 0$ for all vectors $x\in\R^d$.
\end{lemma}
\begin{proof}
Let $P = Z^TA'$ be as defined in \SSh. Consider an arbitrary unit vector $x \in \R^d$, and let $y = A'x$.
\begin{align*}
\|A' x\|^2 - \|P x\|^2 &= \|A'x\|^2 - \|Z^T A'x\|^2 = \|y\|^2 - \|Z^Ty\|^2 \\
&= \|(I-ZZ^T)y\|^2 \geq 0
\end{align*}
and 
\[
\|P x\|^2 - \|B' x\|^2 = \lambda_\ell^2 \sum_{i=1}^\ell \langle x,v_i\rangle^2 \geq 0,
\]
therefore $\|A' x\|^2 - \|B' x\|^2 = (\|A' x\|^2 - \|P x\|^2) + (\|P x\|^2 - \|B' x\|^2) \geq 0$.
\end{proof}

\begin{lemma}
With probability at least $1-\rho_\ell$, Property 2 holds for \SSh: for any unit vector $x\in\R^d$, $\|A'x\|^2 - \|B'x\|^2 \leq 41/16 \; s_\ell^2$. \label{fact1}
\end{lemma}
\begin{proof}
Consider an arbitrary unit vector $x \in \R^d$, and note that
\[
\|A' x\|^2 - \|B' x\|^2 = \left( \|A' x\|^2 - \|P x\|^2  \right) + \left(\|P x\|^2 - \|B' x\|^2 \right).
\]
We bound each term individually.  The first term is bounded as
\begin{align}
\|A' x\|^2 - \|P x\|^2 &= x^T (A'^TA' - P^TP) x &\\
&\leq \|A'^TA' - P^TP\|_2  &\\
&= \|A'^TA' - A'^TZ Z^TA'\|_2 &\\
&= \|A'^T ( I-ZZ^T) A'\|_2 &\\
&= \|A'^T (I-ZZ^T)^T ( I-ZZ^T) A'\|_2 \\
&= \|( I-ZZ^T) A'\|_2^2 &\\
&\leq 25/16\; s_{\ell+1}^2 \leq 25/16\; s_\ell^2.  
\end{align} 
Where transition 5 is true because $(I-ZZ^T)$ is a projection. Transition 7 also holds by the spectral norm error bound of \cite{musco2015stronger} for $\eps=1/4$.
To bound the second term, note that $\|Px\| = \|Z^TA' x\| = \|\Lambda V^T x\|$, since $[H,\Lambda, V] = \svd(P,\ell)$ as defined in \SSh. 
\begin{align*}
&\|P x\|^2 - \|B' x\|^2 = \sum_{i=1}^\ell \lambda_i^2 \langle x,v_i\rangle^2 - \sum_{i=1}^\ell \tilde \lambda_i^2 \langle x,v_i\rangle^2=
\sum_{i=1}^\ell ( \lambda_i^2 - \tilde \lambda_i^2) \langle x,v_i\rangle^2 = \sum_{i=1}^\ell \lambda_\ell^2 \langle x,v_i\rangle^2 \leq \lambda_\ell^2 \leq s_\ell^2,
\end{align*}
where last inequality follows by the Courant-Fischer min-max principle, i.e. as $\lambda_\ell$ is the $\ell$th singular value of the projection of $A'$ onto $Z$, then $\lambda_\ell \leq s_\ell$. Summing the two terms yields $\|A'x\|^2 - \|B'x\|^2 \leq 41/16\; s_\ell^2$.
\end{proof}

The original bound $\|A'^TA' - B'^TB'\|_2 \leq \Delta = 41/8\; s_\ell^2$ discussed in Section~\ref{sec:suc-prob} is also immediately satisfied.

\begin{lemma}
With probability at least $1-\rho_\ell$, Property 3 holds for \SSh: $\|A'\|_F^2 - \|B'\|_F^2 \geq  \ell (3/4) s_\ell^2$.  
\label{lem:prop3}
\end{lemma}
\begin{proof}\[
\|A'\|_F^2 - \|P\|_F^2 = \|A'\|_F^2 - \|Z^TA'\|_F^2 = \|A'-ZZ^TA'\|_F^2 \geq 0
\]
In addition, 
\[
\|P\|_F^2 - \|B'\|_F^2 = \ell \lambda_\ell^2 \geq \ell (3 / 4) s_\ell^2 \ .
\]
The last inequality holds by the per vector error bound of~\cite{musco2015stronger} for $i=\ell$ and $\eps = 1/4$, i.e. 
$|u_\ell^TA'A'^T u_\ell - z_\ell^TA'A'^T z_\ell | = |s_\ell^2 - \lambda_\ell^2| \leq 1/4 s_{\ell+1}^2 \leq 1/4 s_{\ell}^2$, which means $\lambda_\ell^2 \geq 3/4\; s_\ell^2$.
Therefore 
\[
\|A'\|_F^2 - \|B'\|_F^2 = (\|A'\|_F^2 - \|P\|_F^2) + (\|P\|_F^2 - \|B'\|_F^2) \geq \ell (3/4) s_{\ell}^2.  
\]
\end{proof}

\subsubsection{Error Analysis: \BSSh and \sfd}
We now consider the \BSSh algorithm, and the looser version of Property 2 (the original version) as
\begin{itemize} \denselist
\item Property 2: For any unit vector $x\in\R^d$, $\|A'x\|^2 - \|B'x\|^2 \leq \Delta = (41/8) s_\ell^2$.  
\end{itemize}
By invoking $\vs((A'^TA' - B'^TB')/(\Delta/2), \delta)$, then \vs always returns \textsc{True} if 
$\|A'^TA' - B'^TB'\|_2 \leq \Delta/2$ (as is true of the input with probability at least $1-\rho_\ell$ by Lemma \ref{fact1}), and 
\vs catches a failure event where $\|A'^TA' - B'^TB'\|_2 \geq \Delta$ with probability at least $1-\delta_i$ by Lemma \ref{lem:vs}.  
As discussed in Section~\ref{sec:suc-prob} all invocations of \vs succeed with probability at most $1-\delta$, hence all runs of \BSSh succeed and satisfy Property 2 (as well as Properties 1 and 3) with $\alpha = 6/41$ and $\Delta = 41/8\; s_\ell^2$, and with probability at least $1-\delta$.
Finally, we can invoke the mergeability property of $\fd$~\cite{Lib12} and Lemma \ref{lem:desai} to obtain the error bounds in our main result, Theorem \ref{thm:bounds}.

%
%

\section{Experiments}
\label{sec:exp}
In this section we empirically validate that \sfd matches (and often improves upon) the accuracy of \fd, while running significantly faster on sparse real and synthetic datasets.  

We do not implement \sfd exactly as described above.  Instead we directly call \SSh in Algorithm \ref{alg:SparseFD} in place of \BSSh.  The randomized error analysis of \si indicates that we may occasionally miss a subspace within a call of \si and hence \SSh; but in practice this is not a catastrophic event, and as we will observe, does not prevent \sfd from obtaining small empirical error.  

The empirical comparison of \fd to other matrix sketching techniques is now well-trodden~\cite{ghashami2015frequent,desai2015improved}.  \fd (and, as we observe, by association \sfd) has much smaller error than other sketching techniques which operate in a stream.  However, \fd is somewhat slower by a factor of the sketch size $\ell$ up to some leading coefficients.  
We do not repeat these comparison experiments here.  

\Paragraph{Setup} 
We ran all the algorithms under a common implementation framework to test their relative performance as accurately as possible. We ran the experiments on an Intel(R) Core(TM) 2.60 GHz CPU with 64GB of RAM running Ubuntu 14.04.3. 
All algorithms were coded in C, and compiled using gcc 4.8.4.
All linear algebra operation on dense matrices (such as SVD) invoked those implemented in LAPACK. 

\Paragraph{Datasets} 
We compare the performance of the two algorithms on both synthetic and real datasets. Each dataset is an $n\times d$ matrix $A$ containing $n$ datapoints in $d$ dimensions. 

The real dataset is part of the $20$ Newsgroups dataset~\cite{Lang95}, that is a collection of approximately $20{,}000$ documents, partitioned across $20$ different newsgroups. However we use the `by date' version of the data, where features (columns) are tokens and rows correspond to documents. This data matrix is a zero-one matrix with $11{,}314$ rows and $117{,}759$ columns. In our experiment, we use the transpose of the data and picked the first $d = 3000$ columns, hence the subset matrix has $n = 117{,}759$ rows and $d = 3000$ columns; roughly $0.15\%$ of the subset matrix is non-zeros.

The synthetic data generates $n$ rows i.i.d.  
Each row receives exactly $z \ll d$ non-zeros (with default $z=100$ and $d=1000$), with the remaining entries as $0$.  
The non-zeros are chosen as either $1$ or $-1$ at random.  
Each non-zero location is chosen without duplicates among the columns. 
The first $1.5z$ columns (e.g., 150), the ``head'', have a higher probability of receiving a non-zero than the last $d-1.5z$ columns, the ``tail''.   The process to place a non-zero first chooses the head with probability $0.9$ or the tail with probability $0.1$.  For whichever set of columns it chooses (head or tail), it places the non-zero uniformly at random among those columns.  

\Paragraph{Measurements} 
Each algorithm outputs a sketch matrix $B$ of $\ell$ rows.  For each of our experiments, we measure the efficiency of algorithms against one parameter and keep others fixed at a default value.
Table \ref{tbl:param} lists all parameters along with their default value and the range they vary in for synthetic dataset.
We measure the accuracy of the algorithms with respect to:
\begin{itemize}
\item Projection Error:  
\textsf{proj-err} $= \|A - \pi_{B_k}(A)\|_F^2 / \|A-A_k\|_F^2$,
\item Covariance Error: 
\textsf{cov-err} $= \|A^T A-B^T B\|_2 / \|A\|_F^2$,
\item Runtime in seconds. 
\end{itemize}
In all experiments, we have set $k = 10$. Note that \textsf{proj-err} is always larger than $1$, and for \fd and \sfd the \textsf{cov-err} is always smaller than $1/(\frac{6}{41} \ell-k)$ due to our error guarantees.

\begin{table*}[t!!!!]
\caption{Parameter values} \label{tbl:param}
\begin{center}
\begin{small}
\begin{sc}
\begin{tabular}{|c||c|c|}
\hline
 & \textbf{default} & \textbf{range} \\
\hline
\textbf{ datapoints }$(n)$ & $10000$ & $[10^4 - 6\times 10^4]$ \\
\hline
\textbf{dimension }$(d)$ & $1000$ & $[10^3 - 6\times 10^3]$ \\
\hline
\textbf{sketch size }$(\ell)$ & $50$ & $[5-100]$ \\
\hline
\textbf{nnz per row } & $100$ & $[5 - 500]$\\
\hline
\end{tabular}
\end{sc}
\end{small} 
\end{center}
\end{table*}

\begin{table*}
\begin{tabular}{|c|c|c|c|c|} \hline
\rotatebox{90}{\hspace{1mm}  \small \textsf{Projection Error }} &
\includegraphics[width=\figsize]{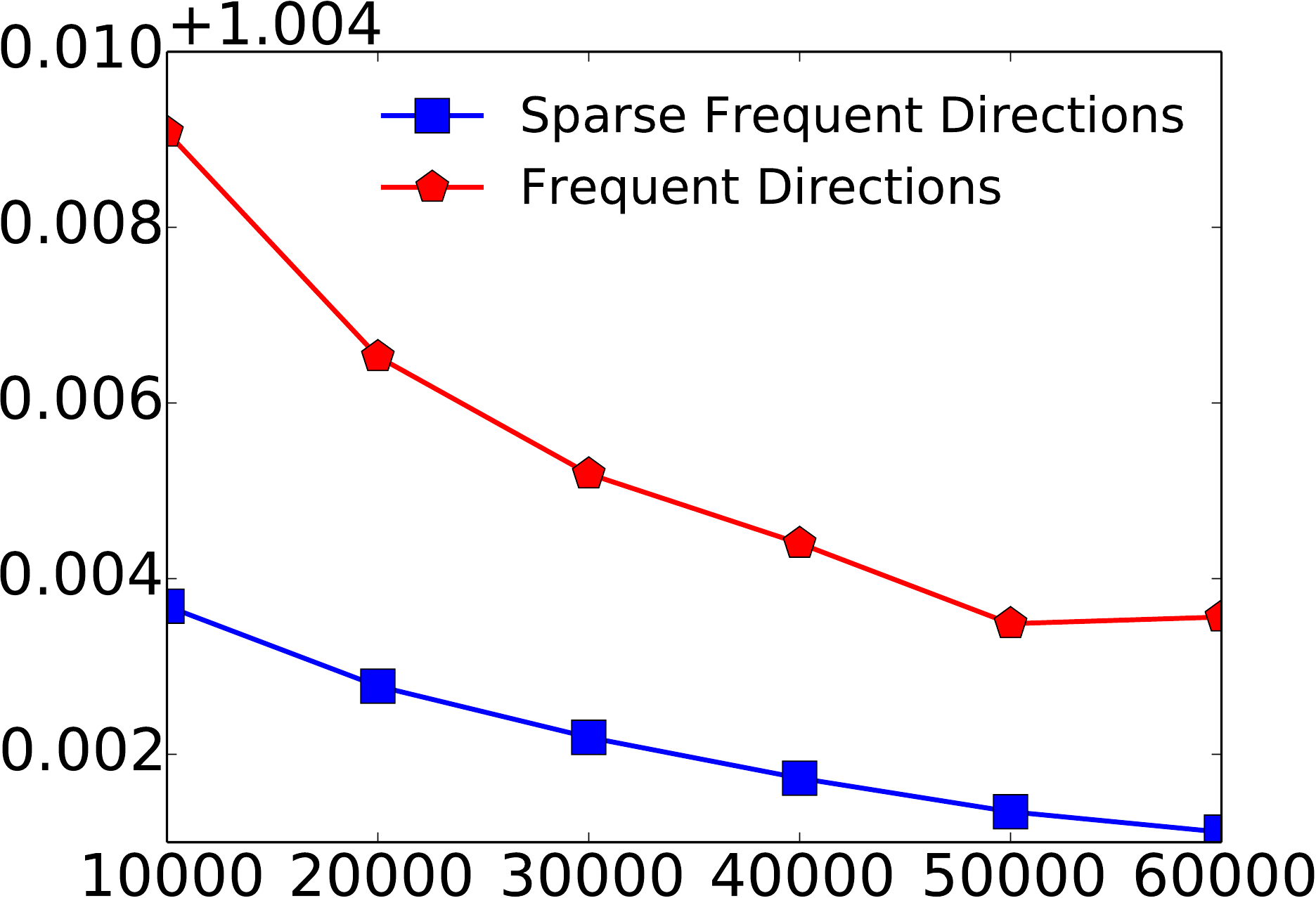} &
\includegraphics[width=\figsize]{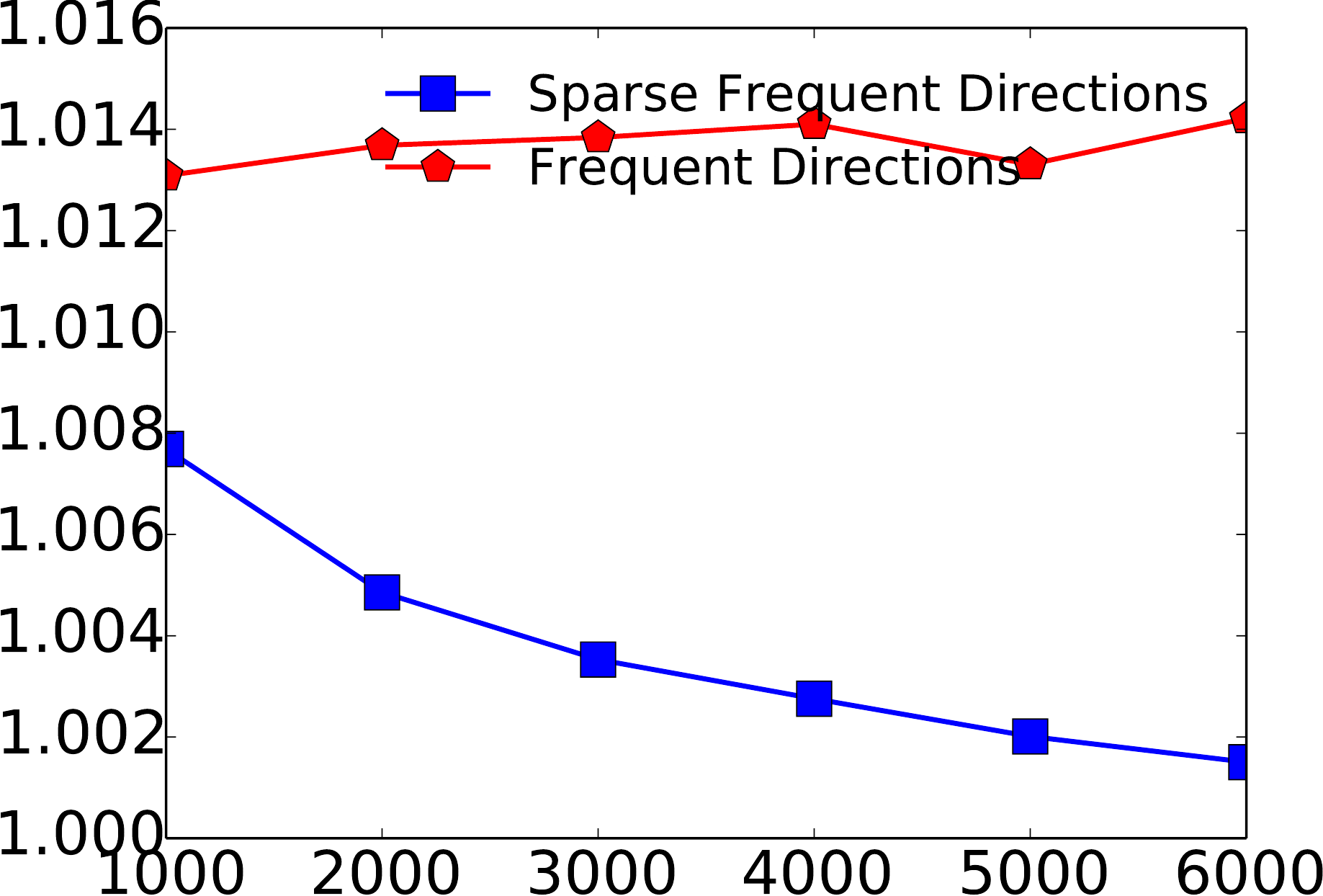} &
\includegraphics[width=\figsize]{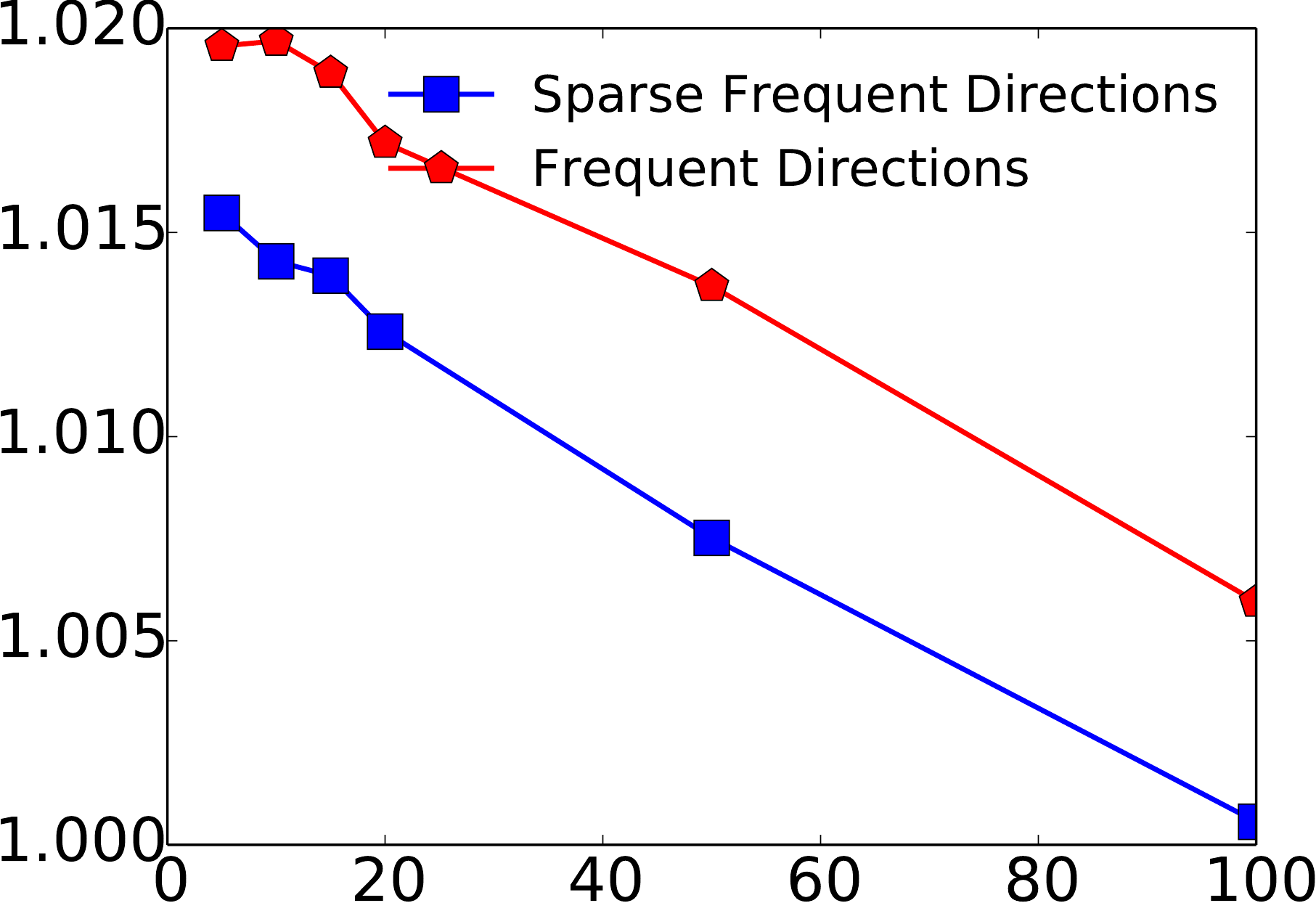} &
\includegraphics[width=\figsize]{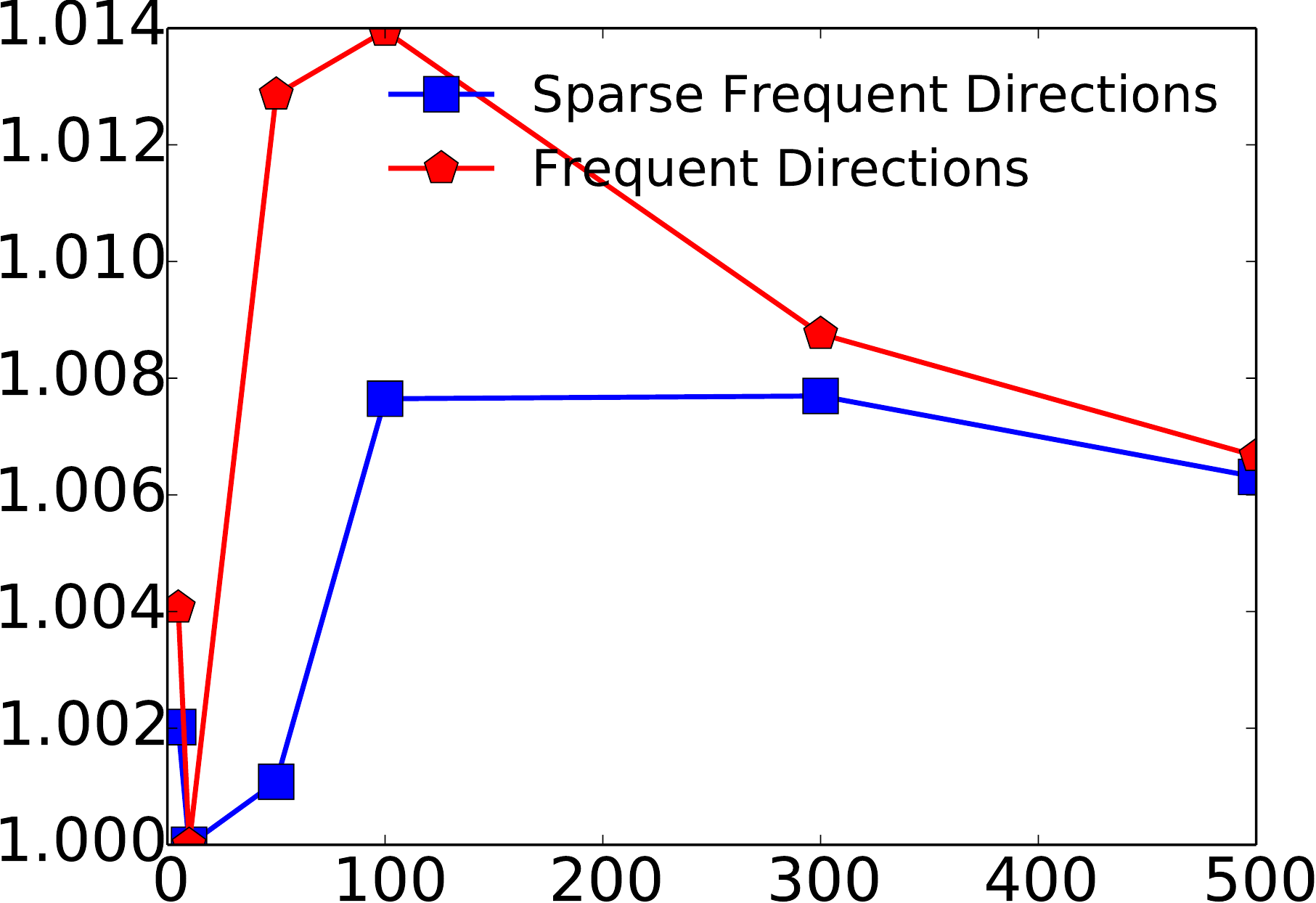} \\ \hline 
\rotatebox{90}{\hspace{1mm}  \small \textsf{Covariance Error }} &
\includegraphics[width=\figsize]{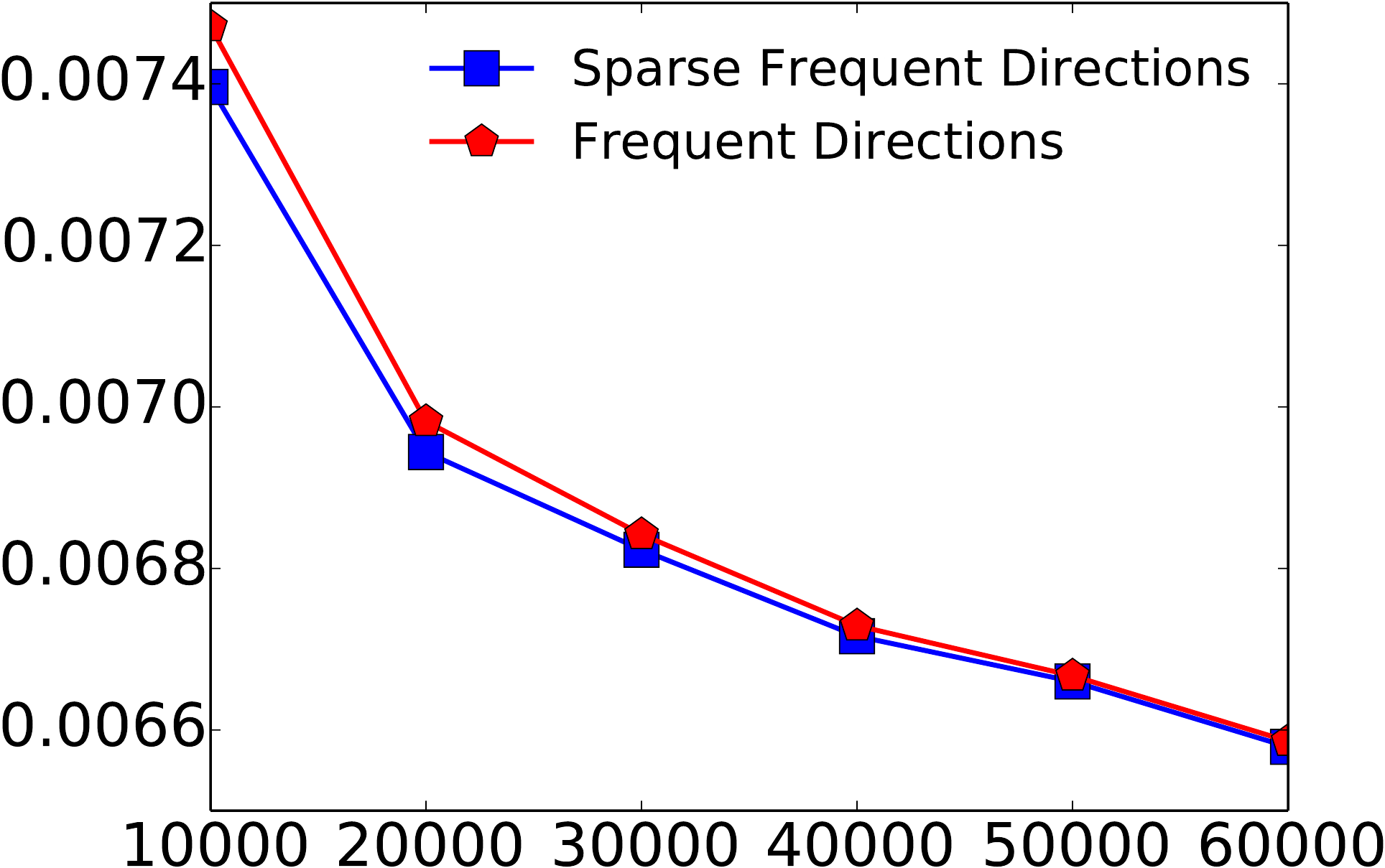} &
\includegraphics[width=\figsize]{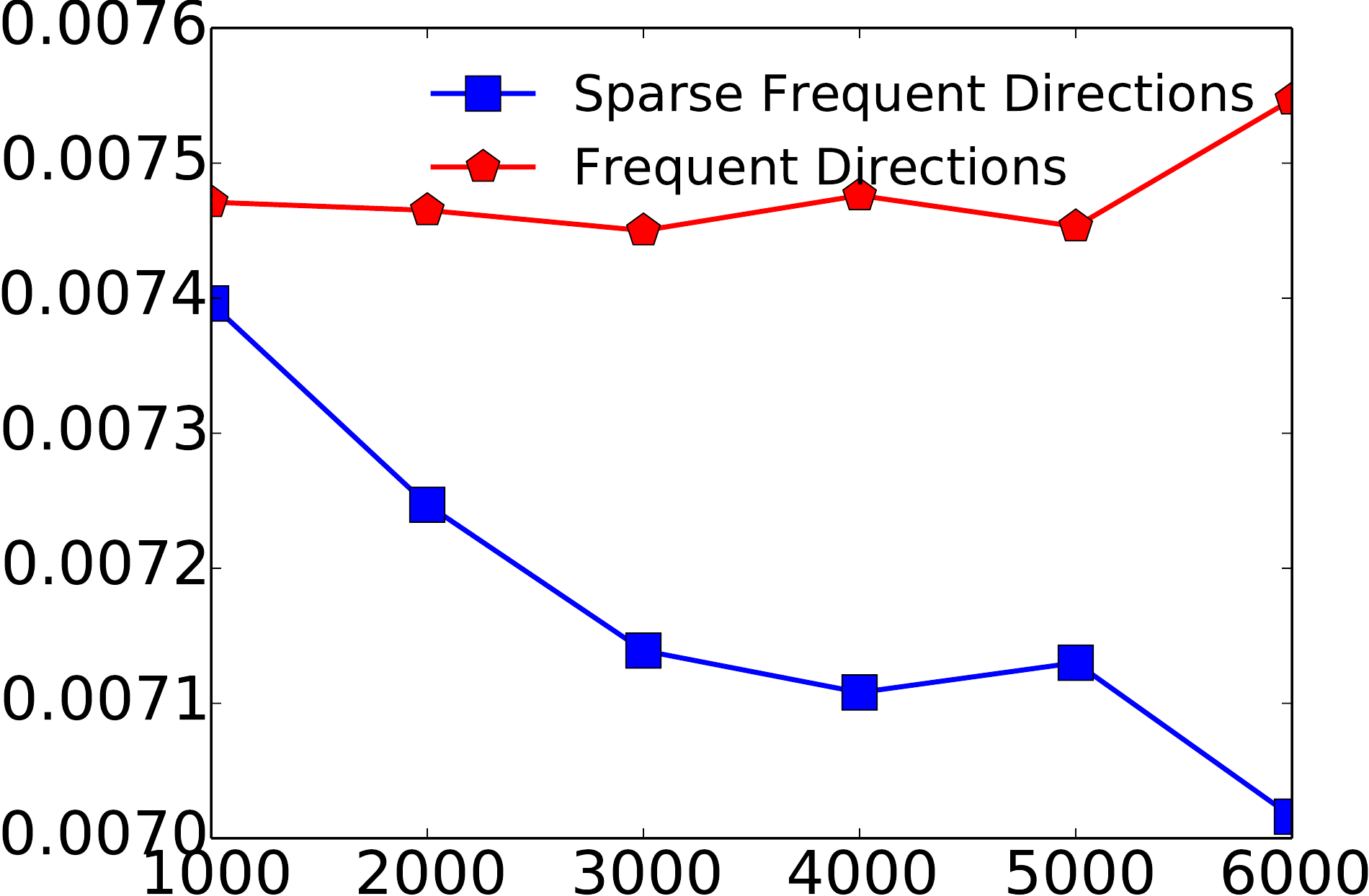} &
\includegraphics[width=\figsize]{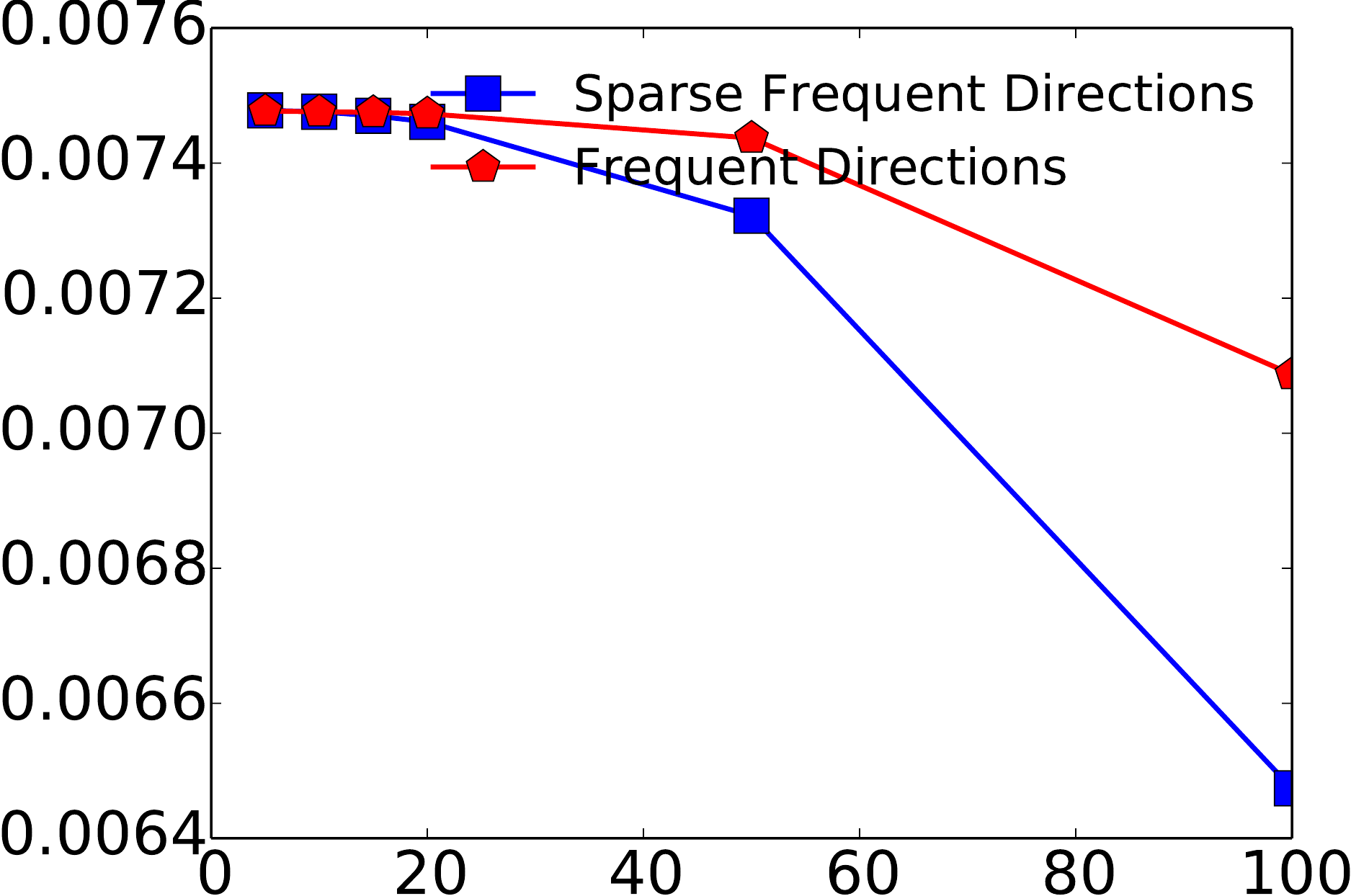} &
\includegraphics[width=\figsize]{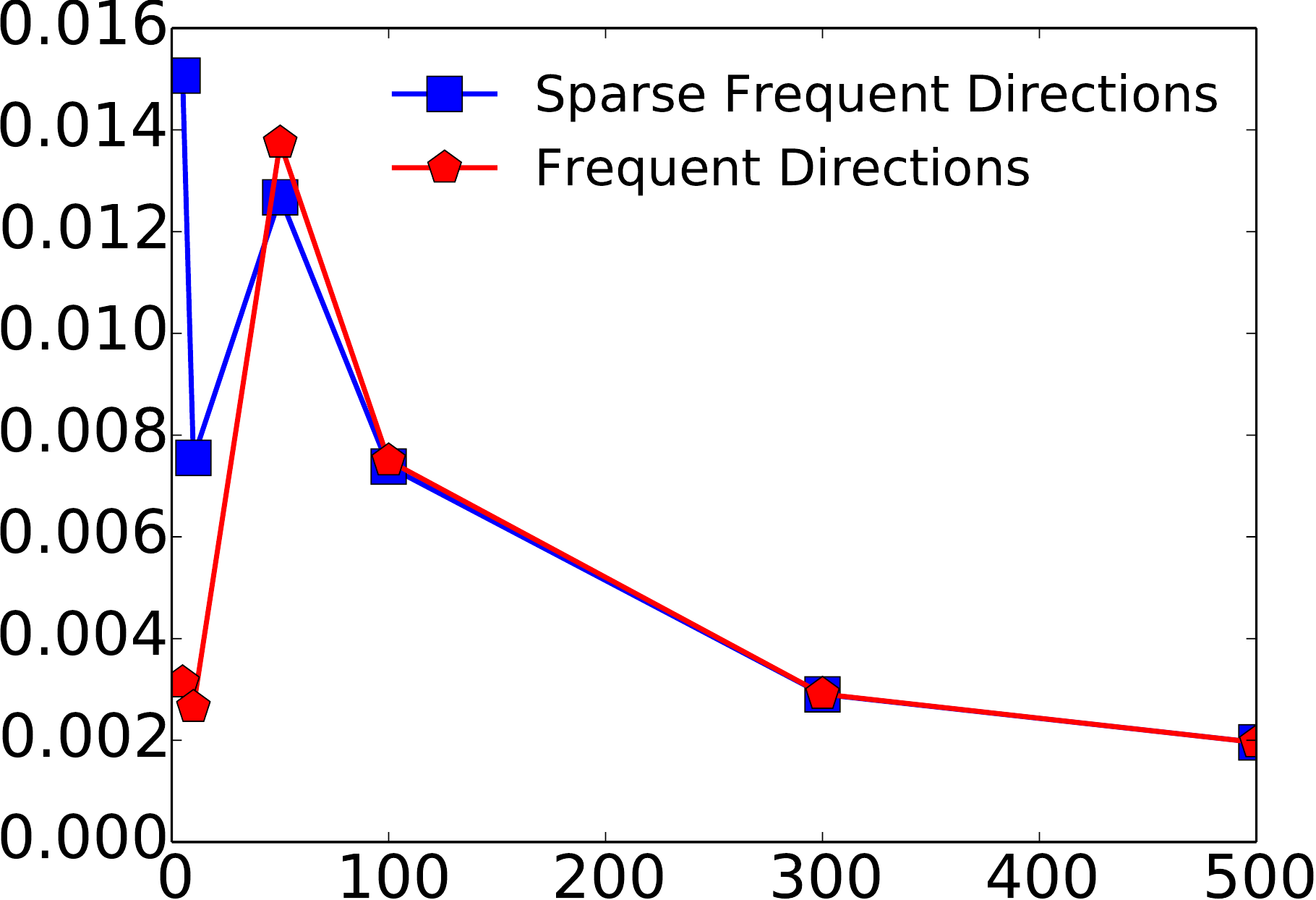} \\ \hline
\rotatebox{90}{\hspace{5mm} \small \textsf{Run Time }} &
\includegraphics[width=\figsize]{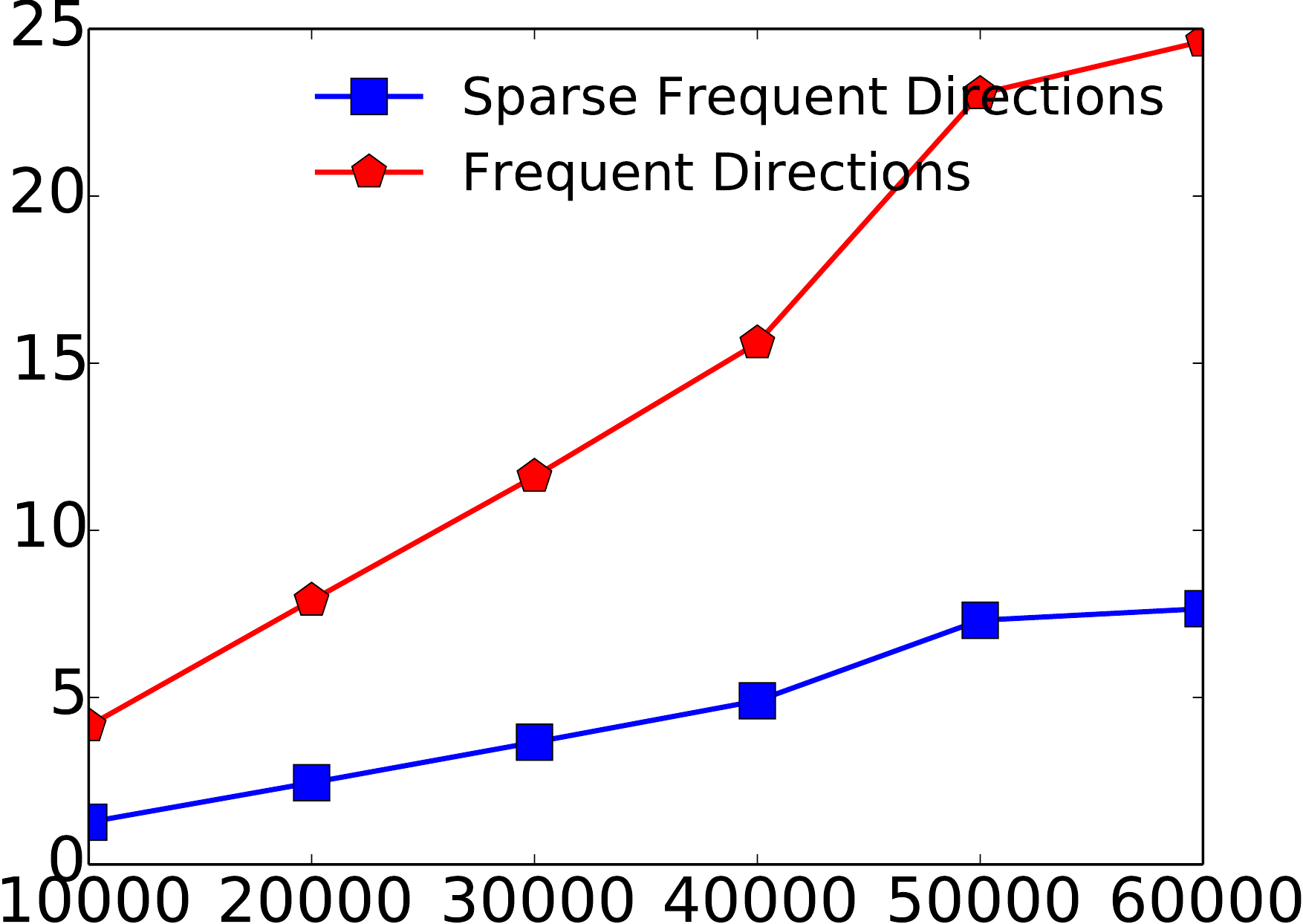} &
\includegraphics[width=\figsize]{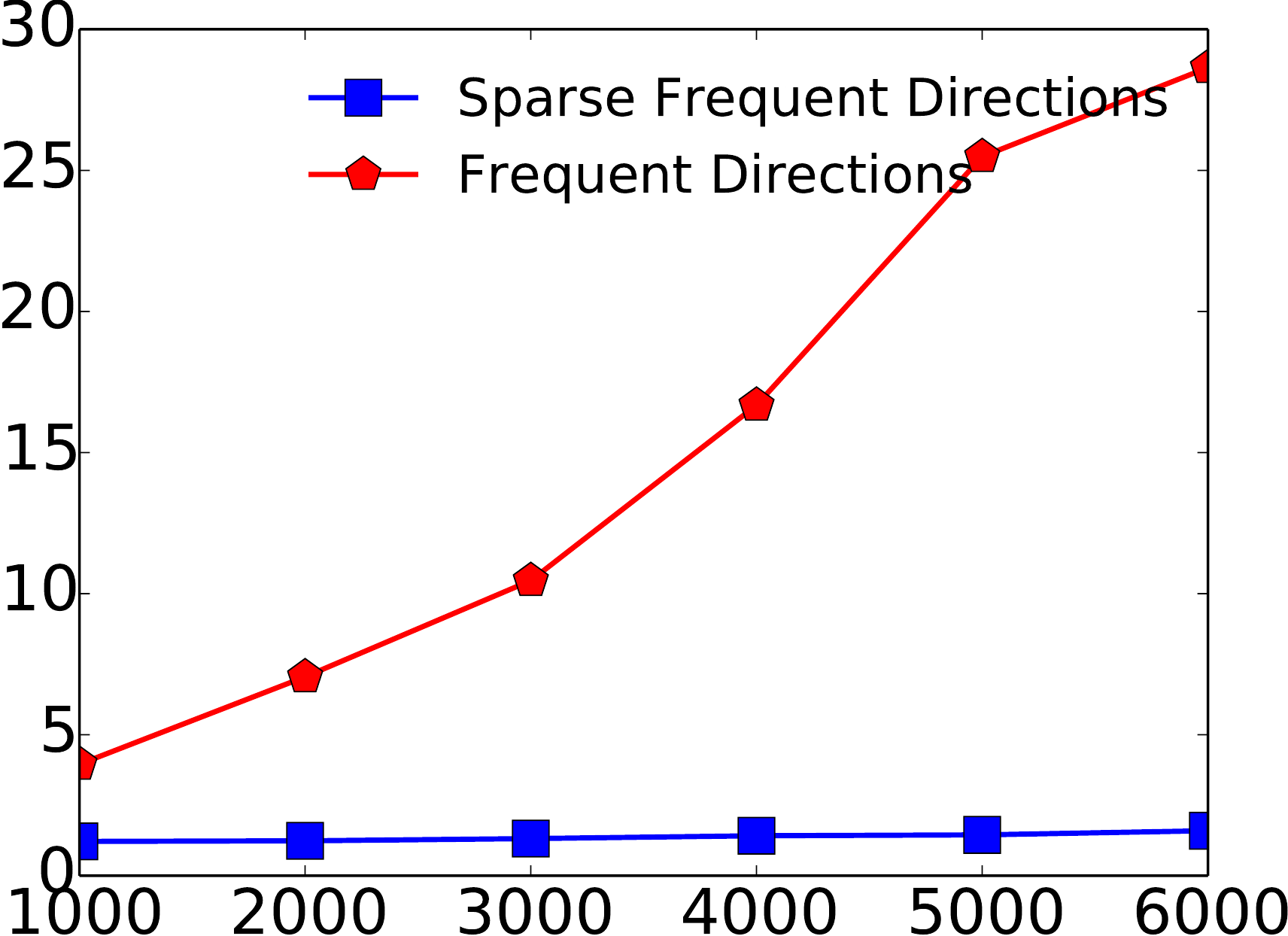} &
\includegraphics[width=\figsize]{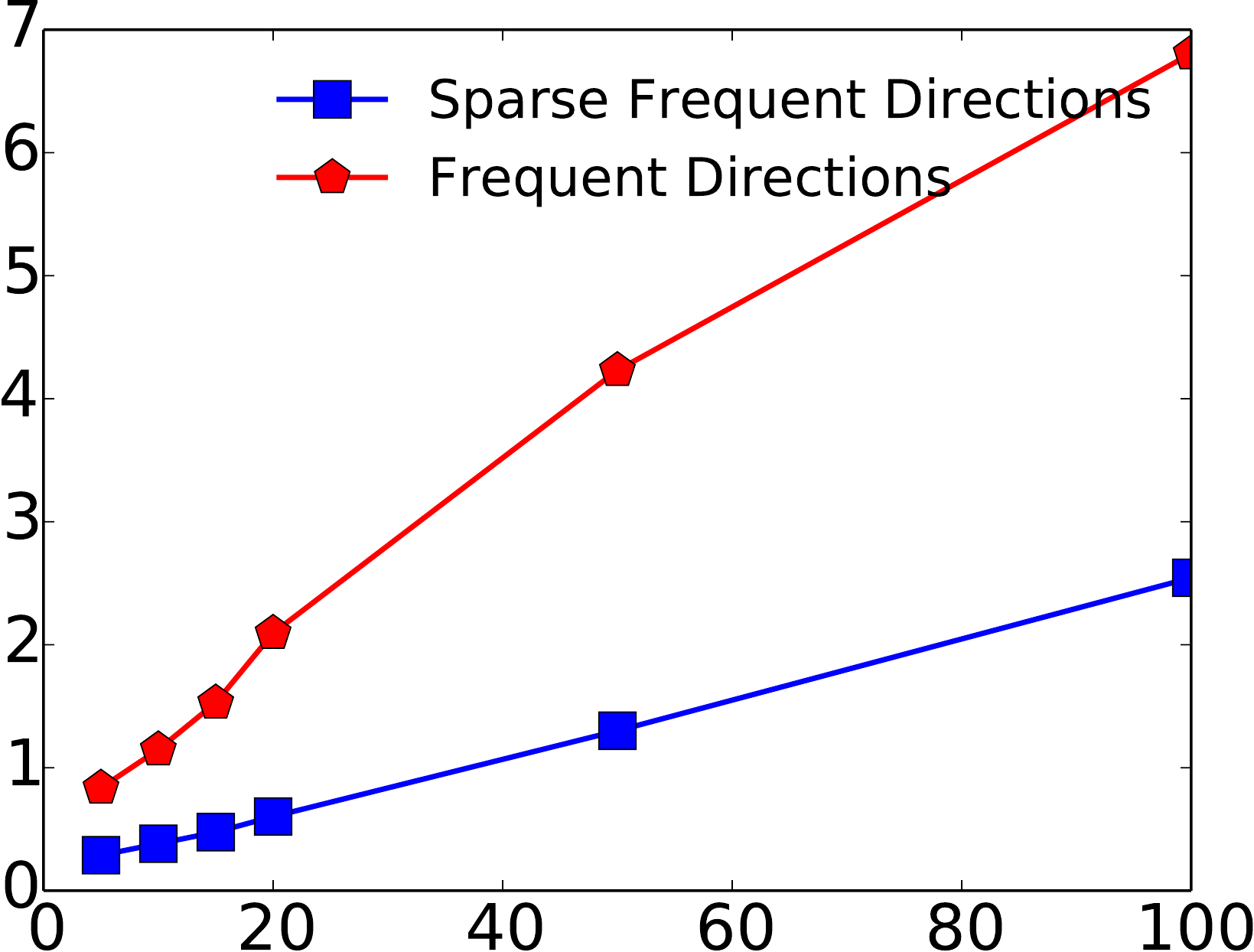} &
\includegraphics[width=0.93\figsize]{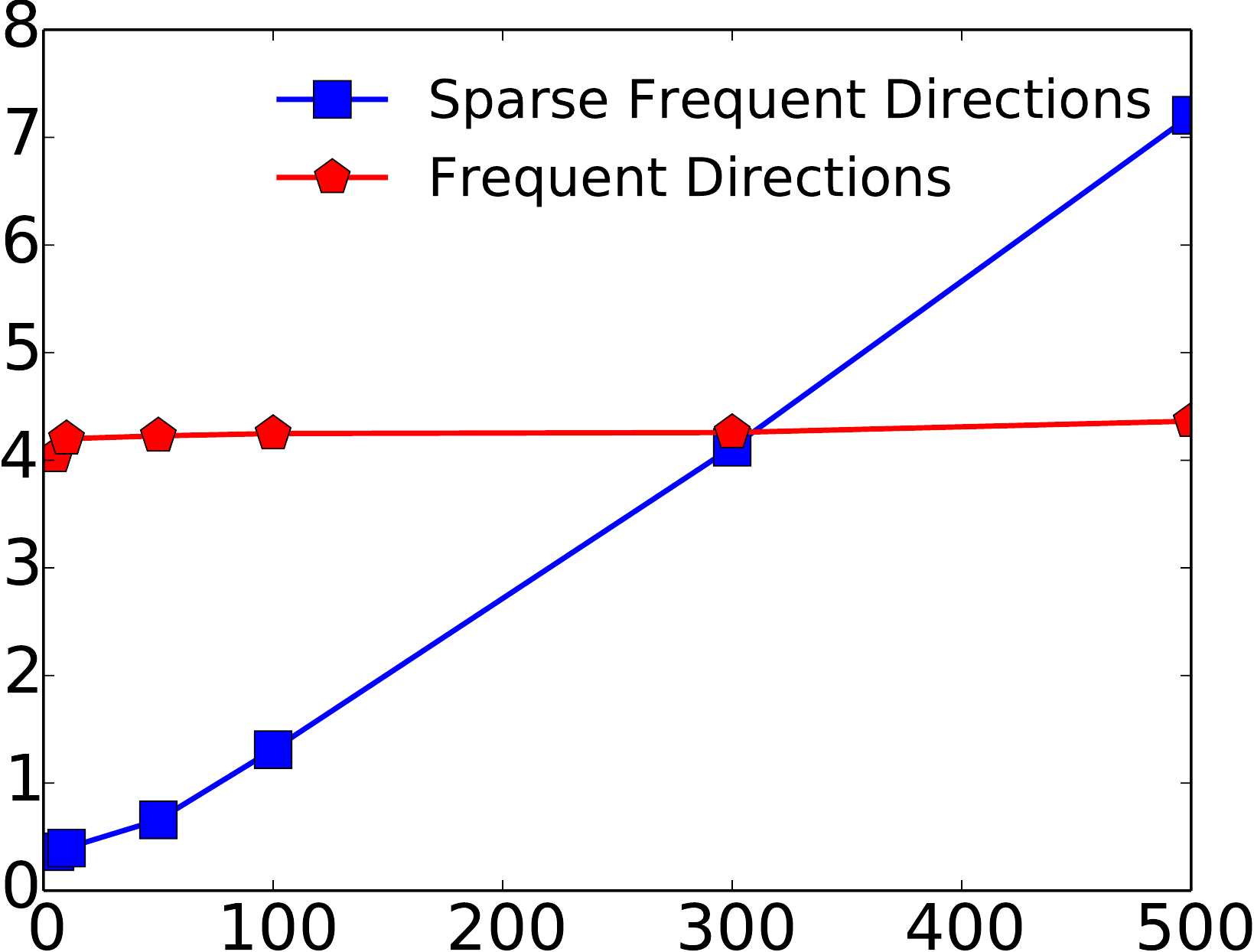} \\ \hline
%
\; &
\textsf{number of data points} &
\textsf{dimension} &
\textsf{sketch size} &
\textsf{nnz per row} \\ \hline
\end{tabular}
\caption{
\label{fig:n}
Comparing performance of \fd and \sfd on synthetic data. Each column reports the measurement against one parameter; ordered from left to right it is number of datapoints $(n)$, dimension $(d)$, sketch size $(\ell)$, and number of non-zeros $(\nnz)$ per row. 
Table \ref{tbl:param} lists default value of all parameters. } 
\end{table*}

%

\begin{figure*}[t!]
\begin{centering}
\includegraphics[width=1.2\figsize]{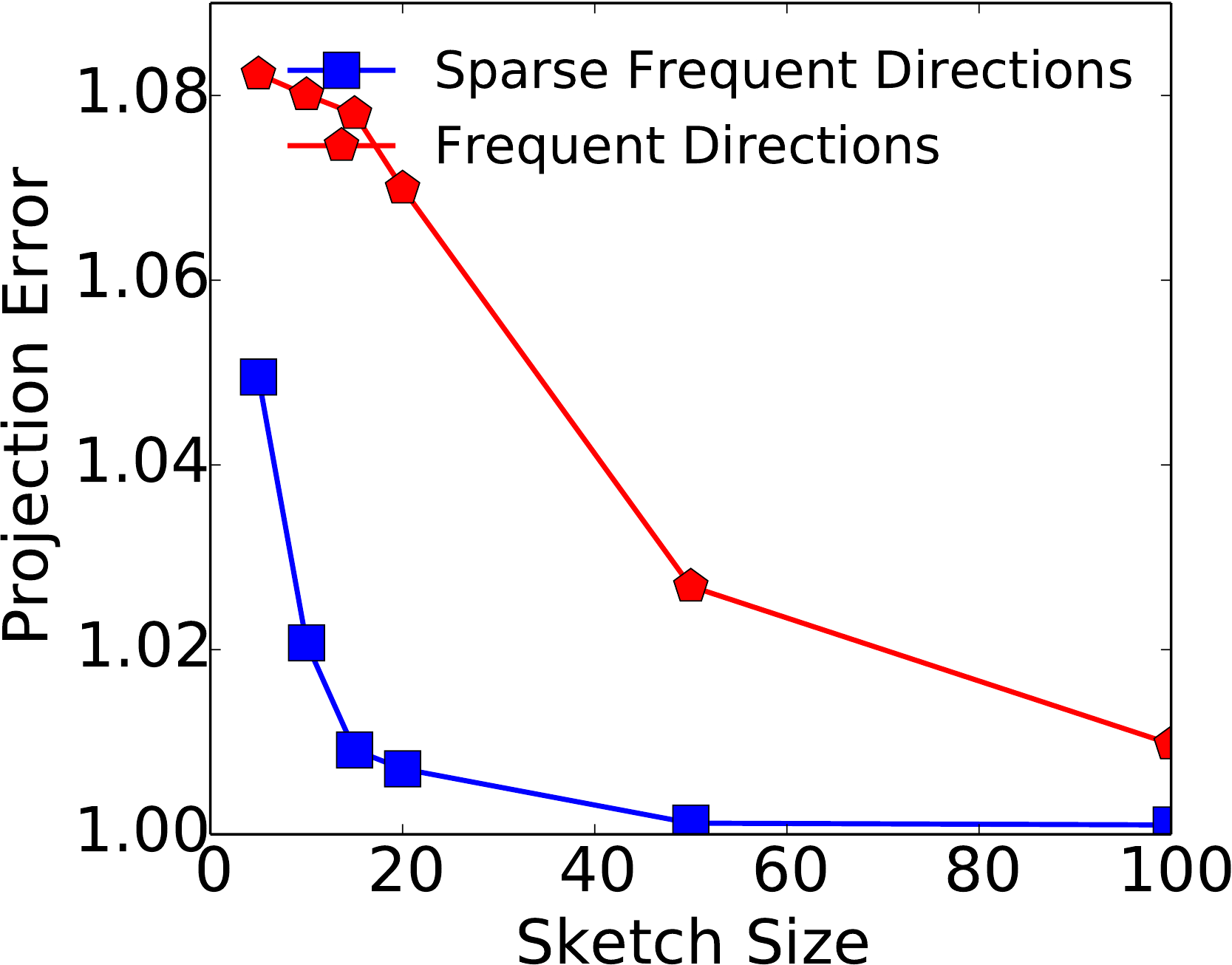} \;\;\;\;\;\;
\includegraphics[width=1.2\figsize]{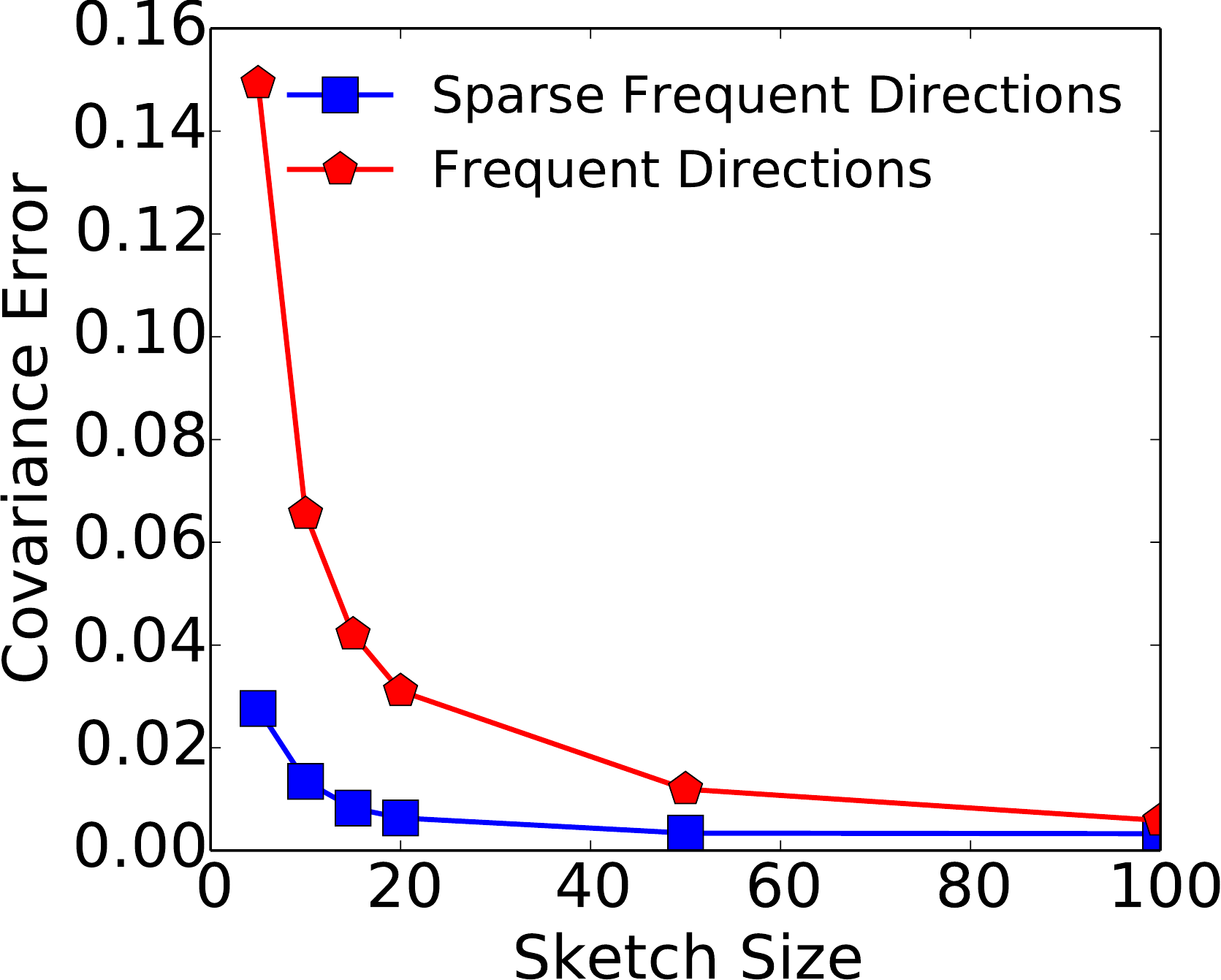} \;\;\;\;\;\;
\includegraphics[width=1.2\figsize]{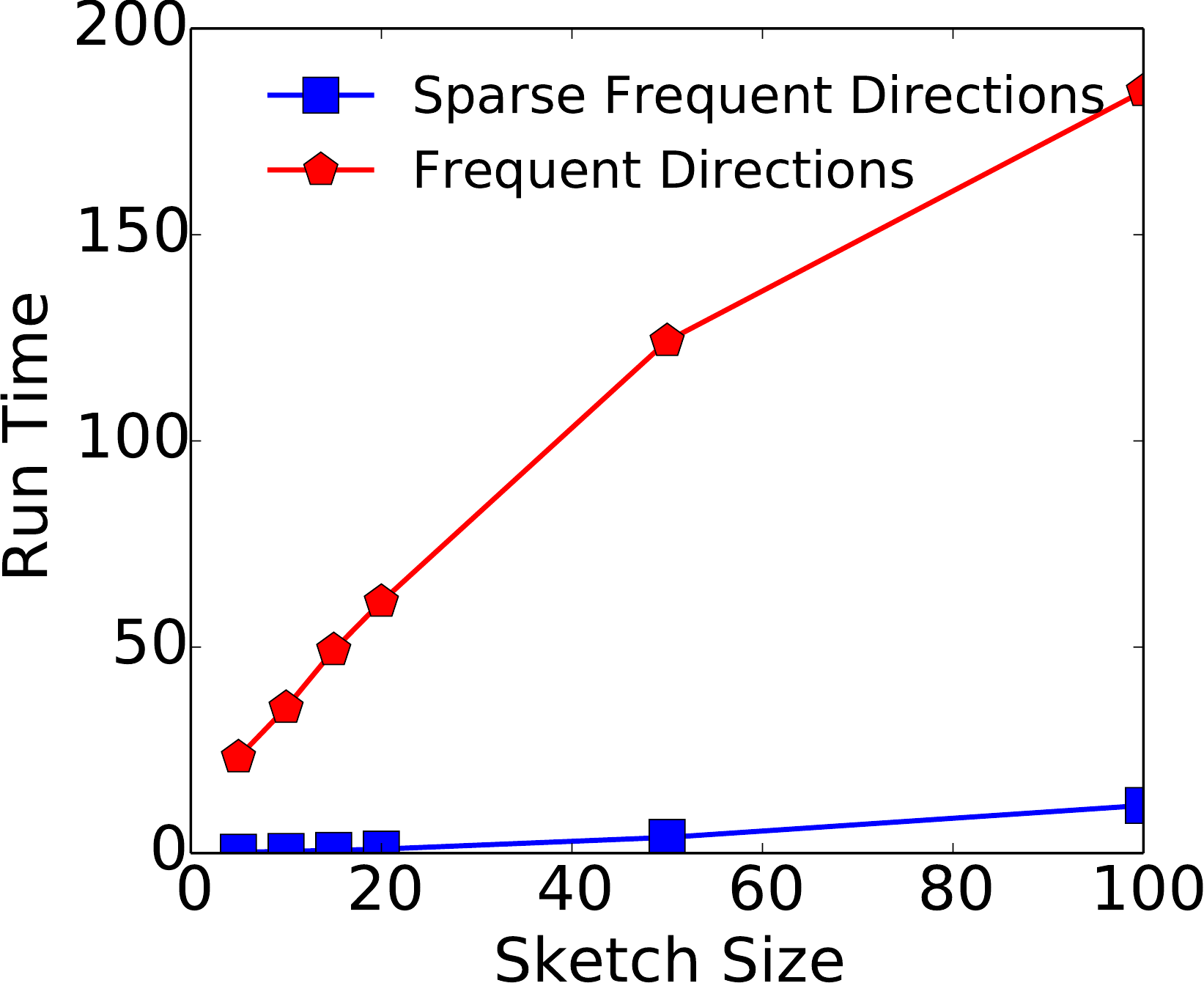}
\caption{
\label{fig:real}
Comparing performance of \fd and \sfd on $20$ Newsgroups dataset.  We plot Projection Error, Covariance Error, and Run Time as a function of sketch size ($\ell$).} 
\end{centering}
\end{figure*}

\subsection{Observations}

By considering Table \ref{fig:n} on synthetic data and Figure \ref{fig:real} on the real data, we can vary and learn many aspects of the runtime and accuracy of \sfd and \fd.  

\Paragraph{Runtime}
Consider the last row of Table \ref{fig:n}, the ``Run Time'' row, and the last column of Figure \ref{fig:real}.  
\sfd is clearly faster than \fd for all datasets, except when the synthetic data becomes dense in the last column of the ``Run Time'' row, where $d=1000$ and $\nnz \text{per row} = 500$ in the right-most data point.  For the default values the improvement is between about a factor of $1.5$x and $2x$, but when the matrix is very sparse the improvement is $10$x or more.  Very sparse synthetic examples are seen in the left data points of the last column, and in the right data points of the second column, of the ``Run Time'' row. 

In particular, these two plots (the second and fourth columns of the ``Run Time'' row) really demonstrate the dependence of \sfd on $\nnz(A)$ and of \fd on $n\cdot d$.  
In the last column, we fix the matrix size $n$ and $d$, but increase the number of non-zeros $\nnz(A)$; the runtime of \fd is basically constant, while for \sfd it grows linearly.  
In the second column, we fix $n$ and $\nnz(A)$, but increase the number of columns $d$; the runtime of \fd grows linearly while the runtime for \sfd is basically constant.  

These algorithms are designed for datasets with extremely large values of $n$; yet we only run on datasets with $n$ up to $60{,}000$ in Table \ref{fig:n}, and $117{,}759$ in Figure \ref{fig:real}.  However, both \fd and \sfd have runtime that grows linearly with respect to the number of rows (assuming the sparsity is at an expected fixed rate per row for \sfd).  This can also be seen empirically in the first column of the ``Run Time'' row where, after a small start-up cost, both \fd and \sfd grow linearly as a function of the number of data points $n$.  Hence, it is valid to directly extrapolate these results for datasets of increased $n$.  

\Paragraph{Accuracy}
We will next discuss the accuracy, as measured in Projection Error in the top row of Table \ref{fig:n} and left plot of Figure \ref{fig:real}, and in Covariance Error in the middle row of Table \ref{fig:n} and middle plot of Figure \ref{fig:real}.  
We observe that both \fd and \sfd obtain very small error (much smaller than upper bounded by the theory), as has been observed elsewhere~\cite{ghashami2015frequent,desai2015improved}.  Moreover, the error for \sfd always nearly matches, or improves over \fd.  We can likely attribute this improvement to being able to process more rows in each batch, and hence needing to perform the shrinking operation fewer overall times.  
The one small exception to \sfd having less Covariance Error than \fd is for extreme sparse datasets in the leftmost data points of Table \ref{fig:n}, last column -- we attribute this to some peculiar orthogonality of columns with near equal norms due to extreme sparsity.

\bibliographystyle{plain}
\bibliography{sparsefd}

\end{document}